\documentclass[copyright,creativecommons]{eptcs}
\providecommand{\event}{GandALF 2020} 
\usepackage{underscore}   

\usepackage{amsmath,amsthm,amssymb}
\usepackage{environ}
\usepackage{stmaryrd}
\usepackage{mathrsfs}
\usepackage{url}
\usepackage{cite}
\usepackage{pifont,xcolor,colortbl}
\usepackage{mleftright}
\usepackage{wrapfig}
\usepackage{cutwin}
\usepackage{floatflt}
\usepackage{wrapfig}
\usepackage{mdwlist}
\usepackage{xspace}
\usepackage{mathtools}

\hypersetup{bookmarks=false}

\usepackage{tikz}
\usetikzlibrary{arrows,automata,positioning,calc}
\usetikzlibrary{shapes,snakes}
\usetikzlibrary{decorations.pathreplacing}
\usepackage[colorinlistoftodos,prependcaption]{todonotes}

\definecolor{Gray}{gray}{0.80}

\DeclareMathAlphabet{\mymathbb}{U}{BOONDOX-ds}{m}{n}

\usepackage{macros}

\title{Synthesis in Presence of Dynamic Links}
\author{B\'eatrice B\'erard$^{1}$,
Benedikt Bollig$^{2}$,
Patricia Bouyer$^{2}$,\\
Matthias F\"ugger$^{2,3}$, and
Nathalie Sznajder$^{1}$\\[1ex]
$^{1}$Sorbonne Universit{\'e}, CNRS, LIP6, F-75005 Paris, France\\
$^{2}$CNRS \& LSV, ENS Paris-Saclay, Universit\'e Paris-Saclay, France\quad $^{3}$Inria, France
}
\def\titlerunning{Synthesis in Presence of Dynamic Links}
\def\authorrunning{B. B\'erard et al.}
\begin{document}
\maketitle

\theoremstyle{plain}
\newtheorem{theorem}{Theorem}
\newtheorem{proposition}{Proposition}
\newtheorem{lemma}{Lemma}
\newtheorem{corollary}{Corollary}
\newtheorem{fact}{Fact}
\theoremstyle{definition}
\newtheorem{definition}{Definition}
\newtheorem{example}{Example}
\theoremstyle{remark}
\newtheorem{claim}{Claim}
\newtheorem{remark}{Remark}

\begin{abstract}
The problem of distributed synthesis is to automatically generate
  a distributed algorithm, given a target communication network
  and a specification of the algorithm's correct behavior.

Previous work has focused on static networks with an a priori fixed
  message size.
This approach has two shortcomings:
Recent work in distributed computing is shifting towards dynamically changing
  communication networks rather than static ones, and an important class
  of distributed algorithms are so-called full-information protocols,
  where nodes piggy-pack previously received messages onto current messages.

In this work, we consider the synthesis problem for a system of two nodes
  communicating in rounds over a dynamic link whose message size is not
  bounded.
Given a network model, i.e., a set of link directions, in each round 
  of the execution, the adversary choses a link from the network model,
  restricted only by the specification, and delivers messages according to the
  current link's directions.
Motivated by communication buses with direct acknowledge mechanisms,
  we further assume that nodes are aware of which messages have been
  delivered.

We show that the synthesis problem is decidable for a network model
  if and only if it does not contain the empty link that dismisses
  both nodes' messages.
\end{abstract}


\section{Introduction}

Starting from Church's work \cite{church1957applications} on synthesizing
  circuits from arithmetic specifications in the 1960s,
  automatic synthesis of programs or circuits
  has been widely studied.

In the case of a reactive system, given a specification, the goal is to find an
  implementation for a system that repeatedly receives inputs from the environment
  and generates outputs such that the system's behavior adheres to the specification.
Early work \cite{rabin1972automata,pnueli1988framework,pnueli1989synthesis} was synthesizing
  algorithms that require knowledge of the complete system state,
  inherently yielding single-process solutions.

Single-process synthesis is related to finding a strategy for a player representing the process
  that has to win against the adversarial environment, and has been studied in the context of
  games~\cite{buchi1990solving,abadi1989realizable,thomas1995synthesis} as well as with automata
  techniques \cite{pnueli1989synthesis,kupferman1999church}.

For systems with more than one process, different models for how
  communication and computation is organized have been studied.
Their two extremes are message-triggered asynchronous computation \cite{gastin2004distributed,MadhusudanTY05}
  and round-wise synchronous computation.

An example for the latter is the work by Pnueli and Rosner~\cite{pnueli1990distributed},
  who considered synchronous distributed systems with an a priori fixed communication
  network.
In their model, the network is given by a directed communication graph,
  whose nodes are the processes and with a link from process $p$ to $q$ if
  $p$ can send messages to $q$ (or write to and read from a shared
  variable).
Messages are from a fixed, finite alphabet per link.
A solution to the synthesis problem is a distributed algorithm that operates in
  rounds, repeatedly reading inputs, exchanging messages, and
  setting outputs.
Already the case of two processes with separate inputs and outputs, and
  without a communication link to each other, was shown to be undecidable
  for linear temporal logic (LTL) specifications~\cite{pnueli1981temporal} on
  the inputs and outputs.
As a positive result, the paper presents a solution for unidirectional process
  chains.

Still in the case of static architectures and bounded messages,
Kupferman and Vardi~\cite{kupferman2000synthesis,kupferman2001synthesizing} extended
  decidability results to branching time specifications and proved sufficient
  conditions on communication networks for decidability,
  while Finkbeiner and Schewe~\cite{finkbeiner2005uniform}
  presented a characterization of networks where synthesis is decidable.
Since specifications are allowed to talk about messages, however, they are powerful enough to break existing communication links between processes, leading to undecidability like in the two-process system without communication~\cite{pnueli1990distributed}.
Gastin \emph{et al.} \cite{gastin2009distributed} proved a necessary and sufficient condition 
  for decidability on a class of communication networks if specifications are only on inputs and
  outputs. Like \cite{gastin2009distributed}, our work only allows ``input-output'' specifications, so that we obtain decidability in several cases where the framework of \cite{finkbeiner2005uniform} does not allow it.

Like in the single-process scenario, synthesis in distributed systems can be modeled
  as a game, which, in this context, are partial information games played between a cooperating
  set of processes against the environment~\cite{peterson1979multiple,mohalik2003distributed,van2005synthesis,berwanger2018hierarchical}. 
With the exception of  \cite{berwanger2018hierarchical},
  all the above approaches assume static, reliable networks.
  In \cite{berwanger2018hierarchical}, Berwanger \emph{et al.}\ study games in which information
  that players have about histories is hierarchically ordered, and this order may change dynamically
  during a play. The main difference to our work is that we consider a memory model
  where messages carry the complete \emph{causal} history
  allowing for unbounded communication messages, while 
\cite{berwanger2018hierarchical} is based on local observations so that,
at every round, a bounded amount of information is transmitted between players.
Further, while asynchronous solutions to the synthesis problem considered
  potentially unbounded messages \cite{MadhusudanTY05,gastin2004distributed},
  previous synchronous solutions
  assume an a priori fixed message size.
Also \cite{MadhusudanTY05} assume that processes that communicate infinitely often
  encounter each other within a bounded number of steps.

The above assumptions have two shortcomings:

\medskip
\noindent \emph{Modeling unreliability.}
Distributed computing has a long history of studying algorithms that provide
  services in presence of unstable or unreliable
  components~\cite{lynch1996distributed}.
Indeed, classical process and link failures can be treated as particular
  dynamic network behavior~\cite{charron2009heard}.
Early work by Akkoyunlu \emph{et al.} \cite{akkoyunlu1975some} considered
  the problem of two groups of gangsters coordinating a coup despite
  an unreliable channel between both parties;
  later on generalized to the Byzantine generals
  problem~\cite{lamport1982byzantine}.
Protocols like the Alternating Bit Protocol~\cite{bartlett1969note}
  aim at tolerating message loss between a sender and receiver node,
  and \cite{aho1982bounds} studies optimal transmission rates over
  unreliable links.
Afek \emph{et al.}~\cite{afek1994reliable} discuss protocols that
  implement reliable links on top of unreliable links.
Further, for algorithms that have to operate in dynamic networks,
  see, e.g., \cite{kuhn2010distributed,coulouma2013characterization,CFN15:icalp},
  network changes are the normal case rather than the
  exception.

Synthesis with unstable or faulty components has been studied by
  Velner and Rabinovich \cite{velner2011church} for two player games in presence
  of information loss between the environment and the inputs of a process.
The approach is restricted to a single process, however.
Dimitrova and Finkbeiner \cite{dimitrova2009synthesis} study synthesis of fault-tolerant
  distributed algorithms in synchronous, fully connected networks.
Processes are partitioned into correct and faulty.
It is assumed that
  at every round at least one process is correct and
  the output of a correct process must not depend on the local inputs of
  faulty processes.
While unreliable links can be mapped to process failures,
  the above assumptions are a priori too restrictive to cover
  dynamic networks.

\medskip

\noindent \emph{Modeling full-information protocols.}
An important class of distributed algorithms are full-information
  protocols, where nodes piggy-pack previously received messages onto
  current messages~\cite{lynch1996distributed,fagin2003reasoning}.
By construction, such algorithms do not have bounded message size.
This kind of causal memory has been considered in
\cite{gastin2004distributed,MadhusudanTY05,GenestGMW13,gimbert18} for synthesis
and control of Zielonka automata over Mazurkiewicz traces with
various objectives, ranging from local-state reachability to $\omega$-branching
behaviors. Zielonka automata usually model
\emph{asynchronous processes} (there is no global clock so that processes
 evolve at their own speed until they synchronize) and
\emph{symmetric communication} (whenever processes synchronize, they mutually exchange their complete
history).

\medskip

In this work we consider the synthesis problem for a system of
  two nodes communicating in synchronous rounds, where
  specifications are given as LTL formulas or, more generally,
  $\omega$-regular languages.
The nodes are connected via a dynamic link.
As in \cite{coulouma2013characterization,CFN15:icalp}, a network is
  a set of communication graphs, called \emph{network model}.
A distributed algorithm operates in rounds as in \cite{pnueli1990distributed}, 
  with the difference that the communication graph is chosen by
  an adversary per round.
Motivated by communication buses, like the industry standard
  I${}^2$C bus \cite{I2C} and CAN bus \cite{CAN},
  with direct acknowledge mechanisms after message transfers,
  we assume that nodes are aware if messages have been delivered
  successfully. In contrast to the Pnueli-Rosner setting, we suppose
  full-information protocols where processes have access to their
  causal history. That is, the dynamic links have unbounded message size.
  Unlike in Zielonka automata over traces, however,
  we consider \emph{synchronous processes} and potentially \emph{asymmetric communication}.
  In particular, the latter implies that a process may learn all about the
  other's history without revealing its own.
  Observe that, when restricting to Zielonka automata, synthesis of asynchronous
  distributed systems is \emph{not} a generalization of the synchronous case.

We show that the synthesis problem is decidable for a network model if and only if
  it does not contain the empty link that dismisses both nodes' messages.
As we assume that LTL specifications can not only reason about inputs and outputs,
but also about the communication graph, our result covers synthesis for dynamic systems
  where links change in more restricted ways.
In particular, this includes processes that do not send further messages after their message
  has been lost, bounded interval omission faults, etc.

\medskip

\noindent \emph{Outline.}
  We define the
  synthesis problem for the dynamic two-process model in
  Section~\ref{sec:problem}.  In Section~\ref{sec:onedirectional-model},
  we discuss the asymmetric model where communication to
  process $1$ never fails. Central to the analysis
  is to show that, despite the availability of unbounded communication
  links, finite-memory distributed algorithms actually suffice.  We
  then prove that the synthesis problem is decidable
  (Theorem~\ref{thm:decidable-asymm}).  In Section~\ref{sec:reduction}
  we reduce the general case of dynamic communication to the
  asymmetric case, obtaining our main result of decidability in
  network models that do not contain the empty link
  (Theorem~\ref{thm:decidable}).  We conclude in
  Section~\ref{sec:conclusion}.
Missing proofs can be found in the long version of the paper
\cite{abs-2002-07545}.


\section{The Synthesis Problem}
\label{sec:problem}

We start with a few preliminaries.
Let $\N = \{0,1,2,\ldots\}$.
For a (possibly infinite) alphabet $A$, the set of finite words over $A$ is denoted by $A^\ast$,
the set of nonempty finite words by $A^+$,
and the set of countably infinite words by $A^\omega$.
We let $\epsilon$ be the empty word and denote
the concatenation of $w_1 \in A^\ast$ and $w_2 \in A^\ast \cup A^\omega$ by $w_1 \cdot w_2$ or simply
  $w_1w_2$.

\smallskip

Fix the set of processes
 \ifdefined\Nprocesses
$\Procs = \{\pone,\ptwo, \dots, N\}$.
\else
$\Procs = \{\pone,\ptwo\}$.
\fi
Every process $p \in \Procs$ comes with fixed finite sets $\Inpp{p}$
and $\Outp{p}$ of possible \emph{inputs} and \emph{outputs}, respectively.
We assume there are at least two possible inputs and outputs per
process, i.e., $|\Inpp{p}| \ge 2$ and $|\Outp{p}| \ge 2$.

We consider systems where computation and communication proceed
in rounds.
In round  $\round = 0,1,2,\ldots$, process $p \in \Procs$ receives an input
$x_p^\round \in \Inp_p$ and it produces an output $y_p^\round \in \Outp{p}$.
The decision on $y_p^\round$ depends on the knowledge that process $p$ has about the
execution up to round $\round$.
In addition to all local inputs $x_p^0,\ldots,x_p^\round$,
this knowledge can also include inputs of the other process,
which may be communicated through communication links.

Following Charron-Bost \emph{et al.}~\cite{CFN15:icalp}, we consider a
  dynamic communication topology in terms of a \emph{network model}, i.e.,
  a fixed nonempty set $\Nmodel \subseteq \{\emptynet,\leftnet,\rightnet,\leftrightnet\}$
  of potentially occurring communication graphs.
In round $\round$, a graph ${\dummynet^\round} \in \Nmodel$
is chosen non-deterministically with the following intuitive meaning:
\begin{description}
\item[$\textcolor{black}{\emptynet}$] No communication takes place.
The knowledge of process $p$ that determines $y_p^\round$
only includes the knowledge at round $\round-1$ as well as the new input $x_p^\round$.

\item[$\textcolor{black}{\leftnet}$] Process 1 becomes aware
of the whole input sequence $x_2^0 \ldots x_2^\round$
that process 2 has received so far. This includes $x_2^\round$, which is transmitted without delay.
The case $\rightnet$ is analogous.

\item[$\textcolor{black}{\leftrightnet}$] Both processes become aware
of the whole input sequence of the other process.
\end{description}
As discussed in the introduction, the knowledge of process $p$ at
round $\round$ also includes the communication link ${\dummynet^\round}$ at $r$,
  which is therefore common knowledge.

\subsection{Histories and Views}

Let us be more formal.
Recall that we fixed the sets $\Procs$, $\Inpp{p}$, $\Outp{p}$, and $\Nmodel$.
We let $\Signals = \Inpp{1} \times \Nmodel \times \Inpp{2}$ be the set of \emph{input signals}.
For ease of notation, we write $\inpsymb{x_1}{\dummynet}{x_2}$ instead of $(x_1,\dummynetw,x_2) \in \Signals$.
Moreover, for ${\dummynet} \in \Nmodel$, we let $\Sigma_\dummynet = \Inpp{1} \times \{\dummynet\} \times \Inpp{2}$.
A word $w \in \Signals^\ast$ represents a possible \emph{history}, a sequence of signals to which the system has been exposed so far.
For a process $p$, we inductively define the \emph{view} $\viewofseq{p}{w}$ of $p$ on $w$
by replacing inputs that are invisible to $p$ by the symbol $\bot$ (we suppose $\bot \not\in \Inpp{1} \cup \Inpp{2}$).
First of all, let $\viewofseq{1}{\epsilon} = \viewofseq{2}{\epsilon} = \epsilon$.
Moreover, for $u \in \Signals^\ast$:
\begin{align*}
\viewofseq{1}{u\inpsymb{x_1}{\leftrightnet}{x_2}} &~=~ u\inpsymb{x_1}{\leftrightnet}{x_2} &
\viewofseq{2}{u\inpsymb{x_1}{\leftrightnet}{x_2}} &~=~ u\inpsymb{x_1}{\leftrightnet}{x_2}
\\
\viewofseq{1}{u\inpsymb{x_1}{\leftnet}{x_2}} &~=~ u\inpsymb{x_1}{\leftnet}{x_2} &
\viewofseq{2}{u\inpsymb{x_1}{\rightnet}{x_2}} &~=~ u\inpsymb{x_1}{\rightnet}{x_2}
\\
\viewofseq{1}{u\inpsymb{x_1}{\rightnet}{x_2}} &~=~ \viewofseq{1}{u}\inpsymb{x_1}{\rightnet}{\bot} &
\viewofseq{2}{u\inpsymb{x_1}{\leftnet}{x_2}} &~=~ \viewofseq{2}{u}\inpsymb{\bot}{\leftnet}{x_2}
\\
\viewofseq{1}{u\inpsymb{x_1}{\emptynet}{x_2}} &~=~ \viewofseq{1}{u}\inpsymb{x_1}{\emptynet}{\bot} &
\viewofseq{2}{u\inpsymb{x_1}{\emptynet}{x_2}} &~=~ \viewofseq{2}{u}\inpsymb{\bot}{\emptynet}{x_2}
\end{align*}
With this, we let $\Hist{1} = \{\viewofseq{1}{w} \mid w \in \Signals^+\}$ and
  $\Hist{2} = \{\viewofseq{2}{w} \mid w \in \Signals^+\}$ be the sets of possible \emph{views} of processes 1 and 2.

The view $\viewofseq{1}{w}$ is illustrated in Figure~\ref{fig:views} for three different words $w$.
For the history in the middle, we have 
$\viewofseq{1}{\inpsymb{x_1^0}{\leftnet}{x_2^0}
\inpsymb{x_1^1}{\leftnet}{x_2^1}
\inpsymb{x_1^2}{\rightnet}{x_2^2}
\inpsymb{x_1^3}{\rightnet}{x_2^3}
} =
\inpsymb{x_1^0}{\leftnet}{x_2^0}
\inpsymb{x_1^1}{\leftnet}{x_2^1}
\inpsymb{x_1^2}{\rightnet}{\bot}
\inpsymb{x_1^3}{\rightnet}{\bot}
$.

\setlength\arrayrulewidth{0.5pt}

\begin{figure}
\centering
{\normalsize
\begin{tabular}{ccc}
$\begin{array}{|cc|c|}
\hline
\cellcolor{Gray}x_1^0 & \cellcolor{Gray}\emptynet & x_2^0\\
\cellcolor{Gray}x_1^1 & \cellcolor{Gray}\emptynet & x_2^1\\
\cellcolor{Gray}x_1^2 & \cellcolor{Gray}\emptynet & x_2^2\\
\cellcolor{Gray}x_1^3 & \cellcolor{Gray}\emptynet & x_2^3\\
\hline
\end{array}$
~~&~~
$\begin{array}{|ccc|}
\hline
\cellcolor{Gray}x_1^0 & \cellcolor{Gray}\leftnet & \cellcolor{Gray}x_2^0\\
\cellcolor{Gray}x_1^1 & \cellcolor{Gray}\leftnet & \cellcolor{Gray}x_2^1\\
\cline{3-3}
\cellcolor{Gray}x_1^2 & \multicolumn{1}{c|}{\cellcolor{Gray}\rightnet} & x_2^2\\
\cellcolor{Gray}x_1^3 & \multicolumn{1}{c|}{\cellcolor{Gray}\rightnet} & x_2^3\\
\hline
\end{array}$
~~&~~
$\begin{array}{|ccc|}
\hline
\cellcolor{Gray}x_1^0 & \cellcolor{Gray}\leftnet & \cellcolor{Gray}x_2^0\\
\cellcolor{Gray}x_1^1 & \cellcolor{Gray}\leftrightnet & \cellcolor{Gray}x_2^1\\
\cellcolor{Gray}x_1^2 & \cellcolor{Gray}\rightnet & \cellcolor{Gray}x_2^2\\
\cellcolor{Gray}x_1^3 & \cellcolor{Gray}\leftnet & \cellcolor{Gray}x_2^3\\
\hline
\end{array}$
\end{tabular}}
\caption{$\viewofseq{1}{w}$ for some histories $w$; the white part is unknown in the view,
  and replaced by $\bot$. \label{fig:views}}
\end{figure}

\subsection{Linear-Time Temporal Logic}

\newcommand{\inform}[2]{(\mathit{in}_{#1} = #2)}
\newcommand{\outform}[2]{(\mathit{out}_{#1} = #2)}
\newcommand{\linkform}[1]{(\mathit{link} = {#1})}
\newcommand{\nextform}[1]{\mathsf{X}#1}
\newcommand{\futureform}[1]{\mathsf{F}#1}
\newcommand{\globallyform}[1]{\mathsf{G}#1}
\newcommand{\untilform}[2]{#1\mathsf{U}#2}
\newcommand{\LTL}[1]{\textup{LTL}(#1)}
\newcommand{\exec}{e}
\newcommand{\ttrue}{\mathit{true}}
\newcommand{\suffix}[2]{#1_{\ge #2}}

Let $\Outputs = \Out_1 \times \Out_2$ be the set of \emph{output signals}.
An \emph{execution} is a word from $(\Signals \times \Outputs)^\omega$,
which records, apart from the input signals, the outputs at every round.
A convenient specification language to define
the \emph{valid} system executions is \emph{linear-time temporal logic} (LTL)
interpreted over words from $(\Signals \times \Outputs)^\omega$.
The logic can, therefore, talk about inputs, outputs, and communication links at a given position.
Moreover, it has the usual temporal modalities. Formally, the set $\LTL{\Nmodel}$ of LTL formulas is given
by the grammar
%
%
\begin{align*}
\varphi :: =~ & \inform{p}{x} \mid \outform{p}{y} \mid \linkform{\dummynet} \mid & \text{atomic formulas}\\
& \nextform{\varphi} \mid \futureform{\varphi} \mid \globallyform{\varphi}
\mid \untilform{\varphi}{\varphi} \mid & \text{temporal modalities}\\
& \neg\varphi \mid \varphi \vee \varphi \mid \varphi \wedge \varphi \mid \varphi \Longrightarrow \varphi
\mid \varphi \Longleftrightarrow \varphi & \text{Boolean connectives}
\end{align*}
%
%
where $p \in \Procs$, $x \in \Inpp{p}$, $y \in \Outp{p}$, and ${\dummynet} \in \Nmodel$.
Let $\exec = \alpha_0\alpha_1\alpha_2 \ldots $ be an execution with
$\alpha_i \in \Signals \times \Outputs$ for all $i \in \N$ and
$\alpha_0 = \bigl(\inpsymb{x_1^0}{\dummynet^0}{x_2^0},(y_1^0,y_2^0)\bigr)$.
For $\round \in \N$, let $\suffix{\exec}{\round}$ denote its suffix $\alpha_{\round}\alpha_{\round+1}\alpha_{\round+2} \ldots$, i.e., $\exec = \suffix{\exec}{0}$.
Boolean connectives are interpreted as usual.
Moreover:
\begin{center}
$
{\arraycolsep=1.4pt\def\arraystretch{1.1}
\begin{array}{lclclcl}
\exec \models \inform{p}{x} & \textup{ if } & x_p^0 = x &~~~~~&
\exec \models \nextform{\phi} & \textup{ if } & \suffix{\exec}{1} \models \phi\\
\exec \models \outform{p}{y} & \textup{ if } & y_p^0 = y &~~~~~&
\exec \models \futureform{\phi} & \textup{ if } & 
\exists \round \ge 0: \suffix{\exec}{\round} \models \phi\\
\exec \models \linkform{\dummynet} & \textup{ if } & {\dummynet^0} = {\dummynet} &~~~~~&
\exec \models \globallyform{\phi} & \textup{ if } & 
\forall \round \ge 0: \suffix{\exec}{\round} \models \phi\\
\multicolumn{7}{l}{
\exec \models \untilform{\phi}{\psi} ~\textup{ if }~
\exists \round \ge 0:
\bigl(
\suffix{\exec}{\round} \models \psi
~\wedge~ \forall 0 \le \roundb < \round: \suffix{\exec}{\roundb} \models \phi
\bigr)}
\end{array}}$
\end{center}
Finally, we let
$L(\phi)=\{\exec \in (\Signals \times \Outputs)^\omega \mid \exec \models \phi\}$
be the set of executions that satisfy $\phi$.

\begin{remark}
  In general, the sequence of communication graphs in an execution
    is arbitrary from $\Nmodel^\omega$, modeling a highly dynamic network without any
    restrictions on stability, eventual convergence, etc.
  Note that the specification is allowed to speak about the
    communication links along a history, however, with the possibility to restrict
    the behavior of the dynamic network and impose process behavior to depend
    on the network dynamics.
\end{remark}

\newcommand{\Zero}{\mymathbb{0}}
\newcommand{\One}{\mymathbb{1}}

\begin{example}\label{ex:ltl}
Suppose $\Inpp{1} = \Inpp{2} = \Outp{1} = \Outp{2} = \{\Zero,\One\}$ and $\Nmodel = \{\leftnet,\rightnet\}$.
Consider
\begin{align*}
\phi_1 &=
\globallyform{\bigl(\outform{1}{\One}\;\Longleftrightarrow\;\outform{2}{\One}\bigr)}
\\
\phi_2 &=
\globallyform{\futureform{\bigl(\inform{1}{\One} \wedge \inform{2}{\One}\bigr)}}
~\Longleftrightarrow~
\globallyform{\futureform{\bigl(\outform{1}{\One} \wedge \outform{2}{\One}\bigr)}}\\
\psi &=
\bigl(\globallyform{\futureform{\linkform{\leftnet}}} \wedge
\globallyform{\futureform{\linkform{\rightnet}}}\bigr)
~\Longrightarrow~ \phi_1 \wedge \phi_2\,.
\end{align*}
Formula $\phi_1$ says that, in each round, both processes agree on their output.
Formula $\phi_2$ postulates that both processes simultaneously output $\One$ infinitely often if, and only if,
both inputs are simultaneously $\One$ infinitely often.
Finally, $\psi$ requires $\phi_1$ and $\phi_2$ to hold if both communication links
occur infinitely often.
We will come back to these formulas later to illustrate the synthesis problem.
\exend
\end{example}

\subsection{Synthesis Problem}

A \emph{distributed algorithm} is a pair $\strat = (\strat_1,\strat_2)$
of functions $f_1: \Hist{1} \to \Outp{1}$ and $f_2: \Hist{2} \to \Outp{2}$
that associate with each view an output.
Given
  $\inpseq = \signal_0\signal_1\signal_2 \ldots \in \Signals^\omega$, we define
  the execution
  $\outcome{\profile}{\inpseq} = \bigl(\signal_0,(y_1^0,y_2^0)\bigr)\bigl(\signal_1,(y_1^1,y_2^1)\bigr)\ldots \in (\Signals \times \Outputs)^\omega$
where
$y_p^\round = \stratp{p}(\viewofseq{\proc}{\signal_0 \ldots \signal_\round})$.
For a finite word $w \in \Signals^\ast$, we define
$\outcome{\profile}{w} \in (\Signals \times \Outputs)^\ast$ similarly
(in particular, $\outcome{\profile}{\epsilon} = \epsilon$).

Let $L \subseteq (\Signals \times \Outputs)^\omega$ and $\phi \in \LTL{\Nmodel}$.
We say that $\profile$ \emph{fulfills} $L$
(respectively $\phi$) if,
for all $w \in \Signals^\omega$, we have $\outcome{\profile}{w} \in L$
(respectively $\outcome{\profile}{w} \in L(\phi)$).
Moreover, we say that $L$ (respectively $\phi$) is \emph{realizable} if
there is some distributed algorithm that fulfills $L$ (respectively $\phi$).

We are now ready to define our main decision problem:

\begin{definition}\label{def:synthesis}
For a fixed network model $\Nmodel$
(recall that we also fixed $\Procs$,
$\Inpp{p}$, $\Outp{p}$), the \emph{synthesis problem}
$\Synthesis{\Nmodel}$ is defined as follows:
\begin{description}
\item[Input:] $\varphi \in \LTL{\Nmodel}$
\item[Question:] Is $\varphi$ realizable?
\end{description}
\end{definition}

\begin{example}\label{ex:synthesis}
Consider the formulas $\phi_1,\phi_2,\psi$ from Example~\ref{ex:ltl}
over $\Nmodel = \{\leftnet,\rightnet\}$.
We easily see that $\phi_1$ is realizable by the distributed algorithm where
both processes always output $\One$.
However, $\phi_1 \wedge \phi_2$ is \emph{not} realizable: if the communication
link is always $\leftnet$ (an analogous argument holds for $\rightnet$), process 2 has no information about any of the inputs
of process~1. Thus, it is impossible for the processes to agree on their outputs in every round while respecting $\phi_2$.
\end{example}
\begin{wrapfigure}{r}{4.5cm}
\centering
{\normalsize
$\begin{array}{cc|ccc|c}
\cline{3-5}%
\textup{round} & & \multicolumn{3}{c|}{\textup{signal}} &\\
\cline{3-5}%
0 & \Zero & \Zero & \leftnet & \One & \Zero\\
\cline{3-5}%
1 & \Zero & \One & \rightnet & \Zero & \Zero\\
2 & \Zero & \One & \rightnet & \One & \Zero\\
3 & \Zero & \Zero & \rightnet & \Zero & \Zero\\
\cline{3-5}%
4 & \One & \One & \leftnet & \Zero & \One\\
5 & \Zero & \One & \leftnet & \One & \Zero\\
\cline{3-5}%
6 & \One & \Zero & \rightnet & \Zero & \One\\
\cline{3-5}%
7 & \Zero & \Zero & \leftnet & \One & \Zero\\
8 & \Zero & \One & \leftnet & \Zero & \Zero\\
\cline{3-5}%
9 & \Zero & \One & \rightnet & \One & \Zero\\
\cline{3-5}%
\end{array}$}
\caption{Fulfilling $\psi$\label{fig:distralgo}}
\end{wrapfigure}
Finally, formula $\psi$ is realizable.
We can now assume that both $\leftnet$ and $\rightnet$ occur infinitely often.
A sequence of signals can be divided into maximal finite blocks with identical
communication links as illustrated in Figure~\ref{fig:distralgo} for the prefix of an execution.
The distributed algorithm proceeds as follows. By default, both processes ouput $\Zero$,
with the following exception: at the first position of each block, a process outputs $\One$
if, and only if, the preceding block contains a round where both processes
simultaneously received $\One$. Note that this preceding block is entirely
contained in the view of both processes. The algorithm's outputs are illustrated in Figure~\ref{fig:distralgo}.
At rounds 4 and 6, they are $\One$ because the corresponding preceding blocks
contain an input pair of $\One$'s. As every block has finite size,
satisfaction of $\phi_2$ is guaranteed.
\exend

\medskip

It is well known that the synthesis problem is undecidable if processes are not connected:

\begin{fact}[Pnueli-Rosner]\label{fact:pnueli-rosner}
The problem $\Synthesis{\{ \emptynet \}}$ is undecidable.
\end{fact}

One also observes that undecidability of the synthesis problem is upward-closed:

\begin{lemma}\label{lemma:propagation}
Let $\Nmodel_1 \subseteq \Nmodel_2$.
If $\Synthesis{\Nmodel_1}$ is undecidable, then so is
  $\Synthesis{\Nmodel_2}$.
\end{lemma}

Indeed, formula $\varphi_1 \in \LTL{\Nmodel_1}$ is realizable iff
formula $\varphi_2 \in \LTL{\Nmodel_2}$ is realizable 
where we let $\varphi_2 =  \bigl(\globallyform{\bigvee_{\dummynet \in \Nmodel_1} \linkform{\dummynet}}\bigr)
\Longrightarrow \varphi_1$.

\medskip

Therefore, we will now focus on network models that do not contain $\emptynet$. Our main result is the following:

\begin{theorem}\label{thm:decidable}
For a network model $\Nmodel$,
  $\Synthesis{\Nmodel}$ is decidable if and only if
  ${\emptynet} \notin \Nmodel$.
\end{theorem}

The ``only if'' direction follows from
Fact~\ref{fact:pnueli-rosner} and Lemma~\ref{lemma:propagation}.
The rest of the paper is devoted to the proof of the
``if'' direction of Theorem~\ref{thm:decidable}.
We will first consider $\Nmodel = \{\leftrightnet, \leftnet \}$
and then reduce the other cases to this particular network model.
By Lemma~\ref{lemma:propagation}, it is enough to do
this reduction for $\{\leftrightnet, \leftnet, \rightnet \}$.


\newcommand{\Dirl}{D_{\leftnet^+}}
\newcommand{\Dirlr}{D_{\leftrightnet\leftnet^\ast}}
\newcommand{\node}{u}

\section{Finite-Memory Distributed Algorithms for $\Nmodel = \{ \leftrightnet, \leftnet \}$}
\label{sec:onedirectional-model}

In this section, we suppose $\Nmodel = \{ \leftrightnet, \leftnet \}$.
We show that, in this case, synthesis is decidable:

\begin{theorem}\label{thm:decidable-asymm}
The problem $\Synthesis{\{ \leftrightnet, \leftnet \}}$ is decidable
(in 4-fold exponential time).
\end{theorem}

As our setting features a dynamic architecture and unbounded message size
in terms of causal histories,
the proof of the theorem requires some new techniques.
In particular, we cannot apply the information-fork criterion
from \cite{finkbeiner2005uniform}, since our specifications can only
describe the link between the processes, and cannot constrain the contents of the messages.

The proof is spread over the remainder of this section as well as Section~\ref{sec:games}.
It crucially relies on the fact that, for every realizable specification $\phi$,
there is a distributed algorithm with a sort of \emph{finite memory} fulfilling it (as shown in this section).
This allows us to reduce, in Section~\ref{sec:games}, the problem of finding a distributed algorithm to
finding a winning strategy in a decidable
game (that we will call a $(2,1)$-player game thereafter) involving 
two cooperating players, where one player has imperfect information,
and an antagonistic environment.

\begin{remark}\label{rem:leftrightnet}
For the sake of technical simplification, we assume in
Sections~\ref{sec:onedirectional-model} and
\ref{sec:games},
without loss of generality, that
input sequences start with a symbol from $\Signals_\leftrightnet =
\Inpp{1} \times \{\leftrightnet\} \times \Inpp{2}$.
Instead of the original formula $\hat{\phi}$, we then simply take
$\phi=\nextform{\hat{\phi}}$.
That is, we can henceforth consider that
$\Hist{1} = \{\viewofseq{1}{w} \mid w \in \Signals_\leftrightnet\Signals^\ast\}$ and
$\Hist{2} = \{\viewofseq{2}{w} \mid w \in \Signals_\leftrightnet\Signals^\ast\}$, and
that a distributed algorithm $\profile$ \emph{fulfills} $\phi \in \LTL{\Nmodel}$ if,
for all $w \in \Signals_\leftrightnet\Signals^\omega$,
we have $\outcome{\profile}{w} \in L(\phi)$.
\end{remark}

\subsection{Finite-Memory Distributed Algorithms}

\subparagraph*{Deterministic Rabin Word Automata.}

Our decidability proof and the definition of a finite-memory distributed algorithm
rely on deterministic Rabin word automata (cf.\ \cite{Thomas90}):

\begin{definition}\label{def:wordaut}
A \emph{deterministic Rabin word automaton (\DRWA)}
over a finite alphabet $A$ is a tuple $\A = (S,\iota,\delta,\Acc)$,
where
$S$ is a finite set of states,
$\iota \in S$ is the \emph{initial state},
$\delta: S \times A \to S$ is the transition function, and
$\Acc \subseteq 2^S \times 2^S$ is the (Rabin) acceptance condition.
\end{definition}

The \DRWA $\A$ defines a language of infinite
words $L(\A) \subseteq A^\omega$ as follows.
We extend~$\delta$ to a function
$\delta: S \times A^\ast \to S$
letting $\delta(s,\epsilon) = s$ and
$\delta(s,aw) = \delta(\delta(s,a),w)$.
Let $w = a_0 a_1 a_2\ldots \in A^\omega$.
We define $\InfVisit{\A}{\iota}{w}= \{s \in S \mid s = \delta(\iota,a_0 \ldots a_i)$ for infinitely many $i \in \N\}$.
We say that $w$ is \emph{accepted} by $\A$ if
there is $(F,F') \in \Acc$ such that
$\InfVisit{\A}{\iota}{w} \cap F \neq \emptyset$ and
$\InfVisit{\A}{\iota}{w} \cap F' = \emptyset$,
i.e., some state of $F$ is visited infinitely often, whereas
all states from $F'$ are visited only finitely often.
We let $L(\A) = \{w \in A^\omega \mid w$ is accepted by $\A\}$.


\subparagraph*{Existence of Finite-Memory Distributed Algorithms.}

We are now ready to state that, if there is a distributed algorithm
that fulfills a specification $\varphi \in \LTL{\Nmodel}$,
then there is also a distributed algorithm
$f$  with finite ``synchronization memory''
in the following sense:
There is a \DRWA $\A$ over $\Signals \times \Omega$
such that
the output of a process for a history $wu$ with
$u \in \Signals_\leftrightnet\Signals_\leftnet^\ast$
only depends on $u$ and the state that
$\A$ reaches after reading $\outcome{\profile}{w}$.
Let $\Signals_{\bot\leftnet} = \{\bot\} \times \{\leftnet\} \times X_2$.

\begin{lemma}\label{lem::finmemory}
Let $\varphi \in \LTL{\Nmodel}$.
There is a \DRWA $\A = (S,\iota,\delta,\Acc)$,
  with $\delta: S \times (\Signals \times \Outputs) \to S$,
  such that
the following are equivalent:
\begin{enumerate}
\item[(1)] There is a distributed algorithm $\strat = (\strat_1,\strat_2)$ that fulfills $\varphi$.

\item[(2)] There is a distributed algorithm $\profile =
  (\strat_1,\strat_2)$ that fulfills $\varphi$ and such that,
for all words $w,w' \in \{\epsilon\} \cup \Sigma_\leftrightnet\Sigma^\ast$ satisfying
$\delta(\iota,\outcome{\profile}{w}) =
\delta(\iota,\outcome{\profile}{w'})$,
the following hold:
\begin{itemize}
\item $\strat_1(wu) = \strat_1(w'u)$ for all $u \in \Signals_\leftrightnet\Signals_{\leftnet}^\ast$

\item $\strat_2(wu) = \strat_2(w'u)$ for all $u \in \Signals_\leftrightnet\Signals_{\bot\leftnet}^\ast$
\end{itemize}
\end{enumerate}
\end{lemma}

\noindent
Note that the acceptance condition and the language of $\A$ are not important in the lemma.

\subsection{Distributed Algorithms as Strategy Trees}
\label{sec:strategy-trees}

Section~\ref{sec:strategy-trees} is devoted to the proof of Lemma~\ref{lem::finmemory}.
The first step is to represent a distributed algorithm as
a \emph{strategy tree}, whose branching structure reflects the algorithm's choices depending on the various inputs.
We then build a tree automaton that accepts a
strategy tree iff it represents a distributed algorithm fulfilling
the given formula $\phi$.
The challenge is to define the tree automaton in such a way that its strategies
  can be cast into hierarchical multiplayer games with \emph{finite sets of observations},
  and that winning strategies within these games are equivalent to distributed algorithms.
We show in this section that this is possible by collapsing potentially unboundedly long
input sequences into an unbounded branching structure.
With this construction, we can show that, if the tree automaton recognizes \emph{some} strategy tree,
  then it also accepts one that represents a finite-memory distributed
  algorithm.

\subparagraph*{Trees and Rabin Tree Automata.}

Let $\Alpha$ be a nonempty (possibly infinite) alphabet and $\Dir$ be a nonempty (possibly infinite)
set of \emph{directions}.
An \emph{$\Alpha$-labeled $\Dir$-tree} is a
  mapping $t: \Dir^\ast \to \Alpha$.
  In particular, $\epsilon$ is the root with label $t(\epsilon)$,
  and $ud$ is the $d$-successor of node $u \in \Dir^\ast$, with label $t(ud)$.

\begin{definition}\label{def:treeaut}
A (nondeterministic) \emph{Rabin tree automaton (\RTA)} over $\Alpha$-labeled $\Dir$-trees is a tuple
$\T = (\tStates,\tinit,\tTrans,\tAcc)$ with
  finite set of states $\tStates$,
  initial state $\tinit \in \tStates$,
  acceptance condition $\tAcc \subseteq 2^\tStates \times 2^\tStates$,
  and (possibly infinite) set of transitions
    $\tTrans \subseteq \tStates \times \Alpha \times \tStates^\Dir$.
\end{definition}

A \emph{run} of $\T$ on an $\Alpha$-labeled $\Dir$-tree $t$ is
  an $\tStates$-labeled $\Dir$-tree $\rho: \Dir^\ast \to \tStates$
  where
  $\rho(\epsilon) = \iota$ (the root is assigned the initial state) and,
  for all $u \in \Dir^\ast$, $\bigl(\rho(u),t(u),d \in \Dir \mapsto \rho(ud)\bigr) \in \tTrans$.
The latter is the transition \emph{applied} at $u$, and we denote it by $\transat{\rho}{u}$.

A path of run $\rho$ is a word $\tpath=d_0d_1d_2 \ldots \in D^\omega$, inducing
the sequence $\epsilon, d_0, d_0d_1, d_0d_1d_2, \ldots$ of nodes visited along $\tpath$.
We let $\Inf(\tpath)$ be the set of states that occur infinitely often as the labels
of these nodes.
Path $\tpath$ is \emph{accepting} if
  there is $(F,F') \in \tAcc$ such that $\Inf(\tpath) \cap F \neq \emptyset$
  and $\Inf(\tpath) \cap F' = \emptyset$.
Run $\rho$ is \emph{accepting} if all its paths are
  accepting.
Finally, $\T$ defines the language of $\Alpha$-labeled $\Dir$-trees
  $L(\T) = \{t: \Dir^\ast \to \Alpha \mid$ there is an accepting run of $\T$ on $t\}$.

\begin{lemma}\label{lem::rationalrun}
Let $\Alpha$ be a singleton alphabet,
  $\Dir$ a nonempty (possibly infinite) set of directions, and
  $\T$ an \RTA over $\Alpha$-labeled $\Dir$-trees (as $A$ is a singleton, we say that $\T$ is \emph{input-free}).
Call a run $\rho$ of $\T$ on the unique $\Alpha$-labeled $\Dir$-tree \emph{rational} if,
  for all $w,w' \in \Dir^\ast$
  with $\rho(w) = \rho(w')$, we have $\transat{\rho}{w} = \transat{\rho}{w'}$.
If $L(\T) \neq \emptyset$, then there is a rational accepting run of $\T$.
\end{lemma}

The lemma essentially follows from
the fact that Rabin games are positionally determined for
the player that aims at satisfying the Rabin objective
\cite{Klarlund94}.
To account for our
non-standard setting of tree automata with possibly infinite $\Dir$,
we give a direct proof in \cite{abs-2002-07545}.

\subparagraph*{Strategy Trees.}

Recall that our goal is to show Lemma~\ref{lem::finmemory} using
strategy trees as a representation of distributed algorithms.
Strategy trees are trees over the (infinite) set of directions $\Dir =
\Signals_\leftrightnet\Signals_\leftnet^\ast$, with the aim to isolate
the positions where a resynchronization occurs, via a letter from
$\Signals_\leftrightnet$. By Remark~\ref{rem:leftrightnet}, we only have 
to consider $\Signals_\leftrightnet\Signals^\ast = (\Signals_\leftrightnet\Signals_\leftnet^\ast)^+=D^+$. 
Hence, to avoid additional notation, we can identify nonempty words in $\Dir^\ast$ with words 
in $\Signals_\leftrightnet\Signals^\ast$. 
It will always be clear from the context whether the underlying
alphabet is $\Dir$ or $\Signals$.

Intuitively, a node $u \in \Dir^\ast$ represents a given history, and
the label of $u$ represents the outputs for possible continuations
from $\Signals_\leftrightnet\Signals_\leftnet^\ast$.  More precisely,
the set $\Pairs$ of labels is the set of pairs
$\tlab=(\tlab_1,\tlab_2)$ where $\tlab_1:
\Signals_\leftrightnet\Signals_\leftnet^\ast \to \Outp{1}$ and
$\tlab_2 : \Signals_\leftrightnet \leftviews^\ast \to \Outp{2}$.  For
$w \in \Signals_\leftrightnet\Signals_\leftnet^\ast$, we define
$\outcome{\tlab}{w} \in (\Signals_\leftrightnet \times \Omega)(\Signals_\leftnet \times \Omega)^\ast$ as expected
(cf.\ the definition of $\outcome{f}{w}$ for a distributed algorithm $f$).
Similarly, for $w \in \Signals_\leftrightnet\Signals_\leftnet^\omega$, we obtain a word
$\outcome{\tlab}{w} \in (\Signals_\leftrightnet \times \Omega)(\Signals_\leftnet \times \Omega)^\omega$.

A \emph{strategy tree} is a $\Pairs$-labeled $\Dir$-tree $t: \Dir^\ast \to \Pairs$.
For $\node \in \Dir^\ast$, let $(\tlab^u_1,\tlab^u_2)$ refer to $t(u)$.
The distributed algorithm associated
with $t$ is denoted by $\profile_{t}$
and is defined as $\profile_{t} = (f_1,f_2)$ as follows
(recall that $\leftviews =
\{\bot\} \times \{\leftnet\} \times \Inpp{2}$):
\begin{itemize}
\item
$f_1(uu') = \tlab_1^{u}(u')$ for all
$u \in \{\epsilon\} \cup \Signals_\leftrightnet\Signals^\ast$ and
$u' \in \Signals_\leftrightnet\Signals_\leftnet^\ast$
\item
$f_2(uu') = \tlab_2^{u}(u')$ for all
$u \in \{\epsilon\} \cup \Signals_\leftrightnet\Signals^\ast$, and
$u' \in \Signals_\leftrightnet\leftviews^\ast$
\end{itemize}

\noindent
In $\tlab_1^{u}(u')$ and $\tlab_2^{u}(u')$, we consider the unique decomposition of $u$ over $D$
so that $f_1$ and $f_2$ are well-defined.

\begin{remark}
The mapping $t \mapsto f_t$ is a bijection. In particular,
for every distributed algorithm $f$, there is a strategy tree $t$
such that $f_t = f$.
\end{remark}

\begin{example}
Suppose $\Inpp{1} = \Inpp{2} = \Outp{1} = \Outp{2} = \{\Zero,\One\}$.
Figure~\ref{fig:strategy-tree} depicts a part of a strategy tree $t$.
Its nodes are gray-shaded.
The labels of nodes of $t$ are themselves represented as (infinite) trees.
Consider the input sequence
$w = \inpsymb{\One}{\leftrightnet}{\One}
\inpsymb{\Zero}{\leftnet}{\Zero}
\inpsymb{\One}{\leftrightnet}{\One}
\inpsymb{\Zero}{\leftnet}{\Zero} \in \Signals_\leftrightnet\Signals^\ast$.
To know what $f_t$ outputs for the first two signals, we look
at the blue-colored nodes of the trees associated with the root of $t$.
To determine the outputs for the two remaining signals,
we look at the red-colored nodes of the trees associated with node $d$.
We thus get
$\outcome{\profile_{t}}{w}
=
(\inpsymb{\One}{\leftrightnet}{\One},(\Zero,\Zero))
(\inpsymb{\Zero}{\leftnet}{\Zero},(\Zero,\One))
(\inpsymb{\One}{\leftrightnet}{\One},(\One,\Zero))
(\inpsymb{\Zero}{\leftnet}{\Zero},(\One,\One))
$
for the whole word $w$.
\exend
\end{example}

\newcommand{\treex}{1.4}
\begin{figure}[t]
\centering
\begin{tikzpicture}[scale=0.82, every node/.style={scale=0.85}]
\tikzset{state/.append style={inner sep=0pt,minimum size=0.8cm}}

\node[state,fill=gray!40] (s0) at (0,0) {$\epsilon$};
\node[state,fill=gray!40] (s0-l) at ($ (s0) + (-1,-4.1) $) {};
\node[state,fill=gray!40] (s0-r) at ($ (s0) + ( 1,-4.1) $) {$d$};
\node[state,fill=gray!40] (s0-rl) at ($ (s0-r) + (-1,-3) $) {$dd$};

\path (s0-rl) -- ++(-0.8,0) node {\large $\dots$};
\path (s0-rl) -- ++( 0.8,0) node {\large $\dots$};

\path (s0-l) -- ++(-0.8,0) node {\large $\dots$};
\path (s0-l) -- (s0-r) node[midway] {\large $\dots$};

\draw[-,thick] (s0) -- (s0-l)
  node[pos=0.47,left]
  {\small $\inpsymb{\Zero}{\leftrightnet}{\Zero}$}
  node[pos=0.6,left]
  {\small $\inpsymb{\One}{\leftnet}{\Zero}$}
  node[pos=0.73,left]
  {\small $\inpsymb{\Zero}{\leftnet}{\One}$};

\draw[-,line width=0.8mm,blue!30!white] (s0) -- (s0-r)
  node[pos=0.53,right,black]
  {\small $\inpsymb{\One}{\leftrightnet}{\One}$}
  node[pos=0.66,right,black]
  {\small $\inpsymb{\Zero}{\leftnet}{\Zero}$};

\draw[-,line width=0.8mm,red!30!white] (s0-r) -- (s0-rl)
  node[pos=0.5,left,black]
  {\small $\inpsymb{\One}{\leftrightnet}{\One}$}
  node[pos=0.7,left,black]
  {\small $\inpsymb{\Zero}{\leftnet}{\Zero}$};

\draw [decorate,decoration={brace,amplitude=2pt,raise=3pt},yshift=0pt]
(1.7,-1.9) -- ++ (0.25,-0.9)
node [black,midway,xshift=0.5cm,yshift=0.1cm]
{\small $d$};

\draw[-,dashed] (s0) -- ++(10,0);

\node[state,fill=white] (lambda1e) at ($ (s0) + (4.8,0) $) {$\lambda^\epsilon_1$};
\node[state,fill=white] (lambda1e-l) at ($ (lambda1e) + (-\treex,-1.5) $) {$\One$};
\node[state,fill=blue!30!white] (lambda1e-r) at ($ (lambda1e) + (\treex,-1.5) $) {$\Zero$};
\path (lambda1e-l) -- (lambda1e-r) node[midway] {\large $\dots$};

\node[state,fill=blue!30!white] (lambda1e-rl) at ($ (lambda1e-r) + (-\treex,-1.5) $) {$\Zero$};
\node[state,fill=white] (lambda1e-rr) at ($ (lambda1e-r) + (\treex,-1.5) $) {$\One$};
\path (lambda1e-rl) -- (lambda1e-rr) node[midway] {\large $\dots$};

\draw[-,thick] (lambda1e) -- (lambda1e-l) node[midway,black,fill=white]
  {\small $\inpsymb{\Zero}{\leftrightnet}{\Zero}$};
\draw[-,line width=0.8mm,blue!30!white] (lambda1e) -- (lambda1e-r) node[midway,black,fill=white]
  {\small $\inpsymb{\One}{\leftrightnet}{\One}$};
\draw[-,line width=0.8mm,blue!30!white] (lambda1e-r) -- (lambda1e-rl) node[midway,black,fill=white]
  {\small $\inpsymb{\Zero}{\leftnet}{\Zero}$};
\draw[-,thick] (lambda1e-r) -- (lambda1e-rr) node[midway,black,fill=white]
  {\small $\inpsymb{\One}{\leftnet}{\One}$};

\node[state,fill=white] (lambda2e) at ($ (s0) + (10,0) $) {$\lambda^\epsilon_2$};
\node[state,fill=white] (lambda2e-l) at ($ (lambda2e) + (-\treex,-1.5) $) {$\One$};
\node[state,fill=blue!30!white] (lambda2e-r) at ($ (lambda2e) + (\treex,-1.5) $) {$\Zero$};
\path (lambda2e-l) -- (lambda2e-r) node[midway] {\large $\dots$};

\node[state,fill=blue!30!white] (lambda2e-rl) at ($ (lambda2e-r) + (-\treex,-1.5) $) {$\One$};
\node[state,fill=white] (lambda2e-rr) at ($ (lambda2e-r) + (\treex,-1.5) $) {$\Zero$};

\draw[-,thick] (lambda2e) -- (lambda2e-l) node[midway,black,fill=white]
  {\small $\inpsymb{\Zero}{\leftrightnet}{\Zero}$};
\draw[-,line width=0.8mm,blue!30!white] (lambda2e) -- (lambda2e-r) node[midway,black,fill=white]
  {\small $\inpsymb{\One}{\leftrightnet}{\One}$};
\draw[-,line width=0.8mm,blue!30!white] (lambda2e-r) -- (lambda2e-rl) node[midway,black,fill=white]
  {\small $\inpsymb{\bot}{\leftnet}{\Zero}$};
\draw[-,thick] (lambda2e-r) -- (lambda2e-rr) node[midway,black,fill=white]
  {\small $\inpsymb{\bot}{\leftnet}{\One}$};

\draw[-,dashed] (s0-r) -- ++(9,0);

\node[state,fill=white] (lambda1d) at ($ (s0-r) + (3.8,0) $) {$\lambda^d_1$};
\node[state,fill=white] (lambda1d-l) at ($ (lambda1d) + (-\treex,-1.5) $) {$\Zero$};
\node[state,fill=red!30!white] (lambda1d-r) at ($ (lambda1d) + (\treex,-1.5) $) {$\One$};
\path (lambda1d-l) -- (lambda1d-r) node[midway] {\large $\dots$};

\node[state,fill=red!30!white] (lambda1d-rl) at ($ (lambda1d-r) + (-\treex,-1.5) $) {$\One$};
\node[state,fill=white] (lambda1d-rr) at ($ (lambda1d-r) + (\treex,-1.5) $) {$\Zero$};
\path (lambda1d-rl) -- (lambda1d-rr) node[midway] {\large $\dots$};

\draw[-,thick] (lambda1d) -- (lambda1d-l) node[midway,black,fill=white]
  {\small $\inpsymb{\Zero}{\leftrightnet}{\Zero}$};
\draw[-,line width=0.8mm,red!30!white] (lambda1d) -- (lambda1d-r) node[midway,black,fill=white]
  {\small $\inpsymb{\One}{\leftrightnet}{\One}$};
\draw[-,line width=0.8mm,red!30!white] (lambda1d-r) -- (lambda1d-rl) node[midway,black,fill=white]
  {\small $\inpsymb{\Zero}{\leftnet}{\Zero}$};
\draw[-,thick] (lambda1d-r) -- (lambda1d-rr) node[midway,black,fill=white]
  {\small $\inpsymb{\One}{\leftnet}{\One}$};

\node[state,fill=white] (lambda2d) at ($ (s0-r) + (9,0) $) {$\lambda^d_2$};
\node[state,fill=white] (lambda2d-l) at ($ (lambda2d) + (-\treex,-1.5) $) {$\Zero$};
\node[state,fill=red!30!white] (lambda2d-r) at ($ (lambda2d) + (\treex,-1.5) $) {$\Zero$};
\path (lambda2d-l) -- (lambda2d-r) node[midway] {\large $\dots$};

\node[state,fill=red!30!white] (lambda2d-rl) at ($ (lambda2d-r) + (-\treex,-1.5) $) {$\One$};
\node[state,fill=white] (lambda2d-rr) at ($ (lambda2d-r) + (\treex,-1.5) $) {$\One$};

\draw[-,thick] (lambda2d) -- (lambda2d-l) node[midway,black,fill=white]
  {\small $\inpsymb{\Zero}{\leftrightnet}{\Zero}$};
\draw[-,line width=0.8mm,red!30!white] (lambda2d) -- (lambda2d-r) node[midway,black,fill=white]
  {\small $\inpsymb{\One}{\leftrightnet}{\One}$};
\draw[-,line width=0.8mm,red!30!white] (lambda2d-r) -- (lambda2d-rl) node[midway,black,fill=white]
  {\small $\inpsymb{\bot}{\leftnet}{\Zero}$};
\draw[-,thick] (lambda2d-r) -- (lambda2d-rr) node[midway,black,fill=white]
  {\small $\inpsymb{\bot}{\leftnet}{\One}$};

\end{tikzpicture}
\caption{A strategy tree $t$\label{fig:strategy-tree}.
}
\end{figure}

Now, Lemma~\ref{lem::finmemory} is a
consequence of the following lemma:

\begin{lemma}\label{lem::finmemorytrees}
Let $\varphi \in \LTL{\Nmodel}$.
There is a \DRWA
$\A = (S, \iota,\delta,\Acc)$,
  with $\delta: S \times (\Signals \times \Outputs) \to S$,
  such that
the following are equivalent:
\begin{enumerate}
\item[(1)] There is a strategy tree $t$ such that
$\profile_{t}$ fulfills $\varphi$.

\item[(2)] There is a strategy tree $t$ such that
(a)
$\profile_{t}$ fulfills $\varphi$, and
(b)
for all $w,w' \in \Dir^\ast$ with
$\delta(\iota,\outcome{\profile_{t}}{w}) =
\delta(\iota,\outcome{\profile_{t}}{w'})$,
we have $t(w) = t(w')$.
\end{enumerate}
\end{lemma}

\newcommand{\flag}{\mathit{flag}}
\newcommand{\inpout}{\alpha}

\begin{proof}
Let $\varphi \in \LTL{\Nmodel}$ be the given formula.
We first define $\A$ and then prove its correctness
in terms of the statement of Lemma~\ref{lem::finmemorytrees}
using an \RTA $\T_\phi$ over strategy trees.

\subparagraph*{The \DRWA $\boldsymbol{\A}$.}

It is well known that there is a \DRWA
$\A_\phi = (S_\phi,\iota_\phi,\delta_\phi,\Acc_\phi)$ over
$\Signals \times \Outputs$,
with doubly exponentially many states and exponentially many acceptance pairs,
such that $L(\A_\phi) = L(\phi)$ (cf.\ \cite{VardiW94,Safra88}).
We refer to states of $\A_\phi$ by $\sphi \in S_\phi$.

Starting from $\A_\phi$, we now define the \DRWA
$\A = (S,\iota,\delta,\Acc)$ such that,
for words that contain infinitely many $\leftrightnet$, it is enough
to look at the sequence of states reached by $\A$
right before the $\leftrightnet$-positions
to determine whether the word is in $L(\A_\phi)$ or not.
The idea is to keep track of the set of states that are taken between two $\leftrightnet$-positions.
Accordingly, the set of states is $S = S_\phi \times 2^{S_\phi}$, with initial state $\iota = (\iota_\phi,\emptyset)$.
Concerning the transitions, for $(\sphi,R) \in S$ and $\inpout=(\inpsymb{x_1}{\dummynetw}{x_2},(y_1,y_2)) \in \Signals \times \Outputs$, we let
\[\delta((\sphi,R),\inpout)
=
\begin{cases}
(\delta_\phi(\sphi,\inpout), \{\delta_\phi(\sphi,\inpout)\} \cup R)     & \textup{if } {\dummynetw} = {\leftnet}\\
(\delta_\phi(\sphi,\inpout), \{\delta_\phi(\sphi,\inpout)\})     & \textup{if } {\dummynetw} = {\leftrightnet}\,.
\end{cases}\]
Finally, the acceptance condition is given by ${\Acc} = \{ (G_F,G_{F'}) \mid (F,F') \in \Acc_\phi\}$ where
$G_F = \{(\sphi,R) \in S \mid F \cap R \neq \emptyset\}$ and
$G_{F'} = \{(\sphi,R) \in S \mid F' \cap R \neq \emptyset\}$.

The following claim states that $\A$ is correct wrt.\
executions with infinitely many synchronization points,
while the acceptance condition is looking only at
states reached right before these synchronizing points:

\begin{claim}\label{claim:bigsteps}
Let
$w_0,w_1,w_2,\ldots \in
(\Signals_\leftrightnet \times \Outputs)(\Signals_\leftnet \times \Outputs)^\ast$.
Moreover, let $w = w_0w_1w_2 \ldots$ be the concatenation of all $w_i$.
Set $\vecs_0 = \iota$ and, for $i \in \N$, $\vecs_{i+1} = (\sphi_{i+1},R_{i+1}) = \delta(\iota,w_0 \ldots w_{i})$.
Then,
$w \in L(\A_\phi)
~\Longleftrightarrow~$
the sequence
$\vecs_0, \vecs_1, \vecs_2, \ldots$ satisfies ${\Acc}
~\Longleftrightarrow~ w \in L(\A)$.
\end{claim}

\subparagraph*{The \RTA $\boldsymbol{\cal T_\varphi}$.}

To get finite-memory algorithms,
we will rely on Lemma~\ref{lem::rationalrun}, which is based on
tree automata. In fact, a crucial ingredient of the proof is an \RTA $\cal T_\varphi$
over $\Pairs$-labeled $D$-trees such that
\begin{center}
$L(\T_\varphi) = \{\,t \mid t$ is a strategy tree such that $f_t$ fulfills $\varphi\}$.
\end{center}
It is defined by $\mathcal{T}_\varphi = (S, \iota,\Delta,\Acc)$
where $S$, $\iota$, and $\Acc$ are taken from $\A$, and
$\Delta$ is given by
\begin{center}
$\begin{array}{rcl}
\Delta =
\mleft\{(\vecs = (\sphi,R),\tlab,(\vecs_d)_{d \in \Dir})
~\middle|~
\begin{array}{lr}
\vecs_d = \delta(\vecs,\outcome{\tlab}{d}) \text{ for all } d \in \Signals_\leftrightnet\Signals_\leftnet^\ast & \textup{(T1)}\\[1ex]
\outcome{\tlab}{w} \in L(\A_\phi[\sphi]) \text{ for all } w \in \Signals_\leftrightnet\Signals_\leftnet^\omega & \textup{(T2)}
\end{array}
\mright\}\,.
\end{array}$
\end{center}
Here, $\A_\phi[\sphi] = (S_\phi,\sphi,\delta_\phi,\Acc_\phi)$ is the
automaton $\A_\phi$ but where $\iota_\phi$ has been replaced
by $\sphi$ as the initial state.
While condition (T1) ``unfolds'' $\A$ into the tree structure taking
care of input sequences with infinitely many synchronization points,
condition (T2) guarantees that the distributed algorithm behaves
correctly should there be no more synchronization.

The proof of correctness of $\T_\varphi$, which relies on
Claim~\ref{claim:bigsteps}, can be found in \cite{abs-2002-07545}.

\subparagraph*{Putting It Together.}

We now obtain Lemma~\ref{lem::finmemorytrees} as a corollary from
Lemma~\ref{lem::rationalrun} using $\T_\phi$.

Direction (2) $\Longrightarrow$ (1) is trivial.
Let us show (1) $\Longrightarrow$ (2) and
suppose $L(\T_\varphi) \neq \emptyset$.
Consider the input-free \RTA
$\T_\phi' = (S, \iota,\Delta',\Acc)$
obtained from $\T_\phi$ by replacing the transition relation with
$\Delta' = \{(\stuple,(\stuple_d)_{d \in \Dir}) \mid (\stuple,\tlab,(\stuple_d)_{d \in \Dir}) \in \Delta\}$.
Note that $L(\T_\varphi') \neq \emptyset$. By Lemma~\ref{lem::rationalrun}, there is
an accepting run $\rho$ of $\T_\phi'$ such that, for all $w,w' \in \Dir^\ast$
with $\rho(w) = \rho(w')$, we have $\transat{\rho}{w} = \transat{\rho}{w'}$.
For all transitions $\theta = (\stuple,(\stuple_d)_{d \in \Dir}) \in \Delta'$,
fix $\tlab^\theta \in \Lambda$ such that $(\stuple,\tlab^\theta,(\stuple_d)_{d \in \Dir}) \in \Delta$.
Let $t: \Dir^\ast \to \Lambda$ be the strategy tree defined by $t(w) = \tlab^{\transat{\rho}{w}}$.

We have $t \in L(\T_\varphi)$. Therefore, $f_t$ fulfills $\varphi$,
i.e., (2a) holds. It remains to show (2b).
Let $w,w' \in \Dir^\ast$ with
$\delta(\iota,\outcome{\profile_{t}}{w}) =
\delta(\iota,\outcome{\profile_{t}}{w'})$.
By induction,
we can show that
$\rho(w) =
\delta(\iota,\outcome{\profile_{t}}{w}) =
\delta(\iota,\outcome{\profile_{t}}{w'}) =
\rho(w')$, i.e., $t(w) = t(w')$, which
proves (2b).
Indeed, $\delta(\iota,\outcome{\profile_{t}}{\epsilon}) = \iota = \rho(\epsilon)$ and,
for $u \in \Dir^\ast$ and $d \in \Dir$, we have
$\delta(\iota,\outcome{\profile_{t}}{ud})
=
\delta(\iota,\outcome{\profile_{t}}{u} \cdot \outcome{\tlab^u}{d})
=
\delta(\delta(\iota,\outcome{\profile_{t}}{u}),\outcome{\tlab^u}{d})
=
\delta(\rho(u),\outcome{\tlab^u}{d})
=
\rho(ud)$.
The last equation is by (T1) in the definition of the transition relation
$\Delta$ of $\T_\phi$.
\end{proof}


\section{From Finite-Memory Distributed Algorithms to Games}
\label{sec:games}

\subsection{Games with Imperfect Information}

The existence of finite-memory distributed algorithms shown in
Section~\ref{sec:onedirectional-model} paves the way for
a reduction of the synthesis problem to $(2,1)$-player games
with imperfect information, where two players form
a coalition against an environment in order to fulfill some objective.
The main differences between games and the synthesis problem are
twofold: Games are played in an arena, on a finite set of nodes (or
states), while the input of the synthesis problem is a logical
specification.  More importantly, in a game, communication between
players occurs implicitly, by observing the nodes that are visited.
Hence, communication between players is bounded by the finite nature
of the arena, whereas in the synthesis problem, processes can send an
unbounded amount of information at each communication point.  Recall
that $\Procs = \{1,2\}$ is the set of processes. In the context of
games, however, its elements are referred to as \emph{players}.

\begin{definition}\label{def:game}
  A \emph{$(2,1)$-player game} is
  a tuple
  $\mathcal{G} =
  (V,v_0,\Winning,\EActions,(\Actions_\proc,\Obs_\proc,\obsf_\proc)_{
    \proc \in \Procs},\GTrans)$.
  Here, $V$ is the finite set of nodes
  containing the initial node $v_0 \in V$.  We assume a Rabin winning condition $\Winning \subseteq
  2^V \times 2^V$. Moreover, $\EActions$ is the finite set of actions
  of the environment, $\Actions_\proc$ is the finite set of actions of
  player $\proc$, $\Obs_\proc$ is the finite set of observations of
  $\proc$, and $\obsf_\proc: V \times \EActions \to \Obs_\proc$
  determines what $\proc$ actually observes for a given node and
  environment action. Finally, $\GTrans: V \times \EActions \times
  (\Actions_1 \times \Actions_2) \to V$ is the transition function.
\end{definition}

The game proceeds in rounds $\round \in \N$, the first round
starting in $v_0$. When a round starts in $v \in V$,
the environment first chooses an action $\eact \in \EActions$.
Players 1 and 2 do not see $\eact$, but only
$\obsf_1(v,\eact)$ and $\obsf_2(v,\eact)$, respectively.
Once the players receive these observations,
they simultaneously choose actions
$a_1 \in \Actions_1$ and $a_2 \in \Actions_2$.
The next state is $\GTrans(v,\eact,(a_1,a_2))$, etc.

Accordingly, a \emph{play} (starting from $v_0$) is a sequence
$\pi = (v_0,\eact_0)(v_1,\eact_1) \ldots \in (V \times \EActions)^\omega$ such that,
  for all $\round \in \N$, there is $(a_1,a_2) \in \Actions_1 \times \Actions_2$ such that
  $v_{\round+1} = \GTrans(v_\round,\eact_\round,(a_1,a_2))$.
The observation that a player $\proc$ collects in play $\pi$ until round $\round$ is
defined as $\obsgame{\proc}{(v_0,\eact_0) \ldots (v_\round,\eact_\round)} = \obsf_\proc(v_0,\eact_0) \ldots \obsf_\proc(v_\round,\eact_\round) \in \Obs_\proc^\ast$.
The play is \emph{winning} (for the coalition of players 1 and 2)
  if $v_0v_1v_2 \ldots$ satisfies the Rabin winning condition in the expected manner.

A \emph{strategy} for player $\proc$ is a mapping $g_\proc: \Obs_\proc^+ \to \Actions_\proc$.
A \emph{strategy profile} is a pair $g = (g_1,g_2)$ of strategies.
We say that play $\pi = (v_0,\eact_0)(v_1,\eact_1) \ldots$ is \emph{compatible} with $g$ if,
  for all $\round \in \N$, we have $v_{\round+1} = \GTrans(v_\round,\eact_\round,(a_1^\round,a_2^\round))$ where
  $a_\proc^\round = g_\proc(\obsgame{\proc}{(v_0,\eact_0) \ldots (v_\round,\eact_\round)})$.
Strategy profile $g$ is \emph{winning} if all
  plays that are compatible with $g$ are winning.

The following fact has been shown by Peterson and Reif \cite{peterson1979multiple} for games and corresponds to the
  undecidability result of Pnueli and Rosner \cite{pnueli1990distributed} for two processes without
  communication.

\begin{fact}[Peterson-Reif]
The following problem is undecidable:
Given a $(2,1)$-player game $\mathcal{G}$,
is there a winning strategy profile?
\end{fact}

Therefore, we have to impose a restriction.
It turns out that, when we translate the synthesis problem
for $\Nmodel=\{\leftrightnet,\leftnet\}$ to games in Section~\ref{sec:reduction-games},
player 1 (who corresponds to process 1) will have perfect information.
We say that player $\proc$ has \emph{perfect information} in $\mathcal{G}$
if $\Obs_p = V \times \EActions$ and $\obsf_\proc$ is the identity function.

The following result is by van der Meyden and Wilke~\cite[Theorem 6]{van2005synthesis}
  with a proof in \cite[Theorem~1]{van2005synthesis:TR}.

\begin{fact}[van der Meyden-Wilke]\label{fact:vanMW}
The following problem is decidable:
Given a $(2,1)$-player game $\mathcal{G}$ such that player 1 has perfect information,
is there a winning strategy profile?
\end{fact}

Note that the transition function of our game
is deterministic so that we actually
obtain decidability in exponential time exploiting a standard technique:
We use a small tree automaton to represent the \emph{global} (full information)
winning strategies and another small alternating tree automaton for the
local ones of player 2 that conform with some global strategy.
The alternating automaton can be checked for nonemptiness
in exponential time.

\subsection{Reduction to Games}
\label{sec:reduction-games}

The analogies between synthesis and games
suggest a natural translation of the former into the latter.
However, the crucial difference being the access to histories,
we rely on the fact that certain histories in distributed
algorithms enjoy a finite abstraction. In fact, it is
enough to reveal a bounded amount of information to player 2 at every
environment action from $\Signals_\leftrightnet$.

\begin{lemma}\label{lem::gamereduction}
Let $\varphi \in \LTL{\Nmodel}$ with $\Nmodel=\{\leftrightnet,\leftnet\}$.
We can effectively construct a $(2,1)$-player game $\mathcal{G}_{\phi}$ such that
player 1 has perfect information and the following holds:
There is a distributed algorithm that fulfills $\varphi$ iff
there is a winning strategy profile in $\mathcal{G}_{\phi}$.
\end{lemma}

\begin{proof}
By Remark~\ref{rem:leftrightnet},
input sequences that do not start with a symbol from $\Sigma_\leftrightnet$ are discarded.
Hence, we assume that those sequences are all trivially ``winning'', i.e.,
$(\Signals_{\leftnet} \times \Omega)(\Signals \times \Omega)^\omega \subseteq L(\phi)$.
Let $\A = (S,\iota,\delta,\Acc)$ be the
\DRWA according to
Lemma~\ref{lem::finmemory}. Recall that
$S = S_\phi \times 2^{S_\phi}$, where
$S_\phi$ is taken from $\A_\phi$,
and that the transition function is of the form $\delta: S \times (\Signals \times \Outputs) \to S$.

We construct the game
$\mathcal{G}_{\phi} = (V,v_0,\Winning,\EActions,(\Actions_\proc,\Obs_\proc,\obsf_\proc)_{
\proc \in \Procs},\GTrans)$ as follows. Obviously,
player $1$ corresponds to process $1$ and player $2$ to process $2$.
We simply set $V = S$ and $v_0 = \iota = (\iota_\phi,\emptyset)$,
and $\Winning$ contains, for all $(F_\phi,F'_\phi) \in \Acc_\phi$, the pair
$(F_\phi \times 2^{S_\phi},F_\phi' \times 2^{S_\phi})$.

Moreover, $\EActions = \Signals$,
the idea being that the environment chooses the inputs and the network graph. 
Accordingly, processes $1$ and $2$ choose their outputs so that $\Actions_1 = \Out_1$ and $\Actions_2 = \Out_2$.

Player 1's observations are $\Obs_1 = V \times \Signals$ and we set
$\obsf_1(s,\inpsymb{x_1}{\dummynet}{x_2}) = (s,\inpsymb{x_1}{\dummynet}{x_2})$.
Thus, player 1 has full information.
Player 2's observations are $\Obs_2 = (S \times \Signals_\leftrightnet) \cup \leftviews$
and we set
\begin{center}
$\obsf_2(\vecs,\inpsymb{x_1}{\dummynetw}{x_2}) =
  \begin{cases}
  (\vecs, \inpsymb{x_1}{\leftrightnet}{x_2}) & \textup{ if } {\dummynet} = {\leftrightnet}\\
  \inpsymb{\bot}{\leftnet}{x_2} & \textup{ if } {\dummynet} = {\leftnet}\,.
  \end{cases}$
\end{center}
That is, when the environment chooses a synchronizing input signal,
the current state of $\A$ is revealed to player 2, which corresponds to
passing the (abstracted) history to process 2.
Finally, the transitions are given by
$\GTrans(\vecs,\inpsymb{x_1}{\dummynetw}{x_2},(y_1,y_2)) =
  \delta\bigl(\vecs,(\inpsymb{x_1}{\dummynetw}{x_2},(y_1,y_2))\bigr)$.

The proof of correctness of the reduction is available in \cite{abs-2002-07545}.
\end{proof}

We have shown Theorem~\ref{thm:decidable-asymm} saying that
the problem $\Synthesis{\{ \leftrightnet, \leftnet \}}$ is decidable.

\subparagraph*{Complexity.}
The size of $\A_\phi$ is doubly exponential in the length of the formula.
It follows that the size of $\A$ is triply exponential, and so is
the size of $\mathcal{G}_\phi$.
Deciding the winner of our $(2,1)$-player game where one player has perfect information
can be done in exponential time so that the overall decision
procedure runs in 4-fold exponential time.

Note that $\Synthesis{\{\leftrightnet\}}$, which is equivalent to
centralized synthesis in presence of one single process, is 2EXPTIME-complete
\cite{pnueli1988framework}, from which we inherit the best known lower bound
for our problem. Moreover, hierarchical information further increases the
complexity: for static pipelines with variable number
of processes, the problem is no longer elementary \cite{pnueli1990distributed}.
However, it may be possible to improve our upper bound, which is left for future work.


As, in the proof, the given LTL formula is translated into
a \DRWA, synthesis is decidable
even when the specification is given
by any common finite automaton over $\omega$-words
(starting with a nondeterministic B{\"u}chi automaton,
we actually save one exponential wrt.\ LTL):

\begin{corollary}\label{cor:autsynth}
Over $\Nmodel = \{ \leftrightnet, \leftnet \}$, the following problem is decidable: 
Given an $\omega$-regular language $L \subseteq (\Signals \times \Outputs)^\omega$,
is $L$ realizable?
\end{corollary}


\section{Reduction from $\{ \leftrightnet, \leftnet, \rightnet \}$ to
  $\{ \leftrightnet, \leftnet \}$}
\label{sec:reduction}

\newcommand{\dsignal}{\#}
\newcommand{\doutput}{\#}
\newcommand{\bij}{\beta}
\newcommand{\pword}{u}

\newcommand{\MOD}[1]{\textup{sim}_{#1}}
\newcommand{\simp}[1]{\textup{sim}_{#1}}
\newcommand{\simone}{\textup{sim}_1}
\newcommand{\simtwo}{\textup{sim}_2}

\newcommand{\figtranslation}{
\begin{figure}[h]
\centering
\small
{
$\begin{array}{c}
\\0\\1\\\\2\\3\\4\\5\\\\
\end{array}
\begin{array}{ccc}
w &~~~~& \trans{w}\\
\begin{array}{|ccc|}
\hline
\cellcolor{Gray}x_0 & \cellcolor{Gray}\leftnet & \cellcolor{Gray}x_0'\\
%
%
\cline{3-3}
\cellcolor{Gray}x_1 & \multicolumn{1}{c|}{\cellcolor{Gray}\rightnet} & x_1'\\
\cellcolor{Gray} & \multicolumn{1}{c|}{\cellcolor{Gray}} & \\
\cellcolor{Gray}x_2 & \multicolumn{1}{c|}{\cellcolor{Gray}\rightnet} & x_2'\\
\cline{1-2}
x_3 & \leftrightnet & x_3'\\
x_4 & \rightnet & x_4'\\
x_5 & \leftnet & x_5'\\
& & \\
%
\hline
\end{array}
& \rightsquigarrow &
\begin{array}{|ccc|}
\hline
\cellcolor{Gray}x_0 & \cellcolor{Gray}\leftnet & \cellcolor{Gray}x_0'\\
\cellcolor{Gray}\dsignal & \cellcolor{Gray}\leftrightnet & \cellcolor{Gray}\dsignal\\
\cline{1-1}
x_1' & \multicolumn{1}{|c}{\cellcolor{Gray}\leftnet} & \cellcolor{Gray}x_1\\
x_2' & \multicolumn{1}{|c}{\cellcolor{Gray}\leftnet} & \cellcolor{Gray}x_2\\
\cline{2-3}
x_3' & \leftrightnet & x_3\\
x_4' & \leftnet & x_4\\
\dsignal & \leftrightnet & \dsignal\\
x_5 & \leftnet & x_5'\\
%
\hline
\end{array}
\end{array}
\begin{array}{c}
\\0\\1\\2\\3\\4\\5\\6\\7
\end{array}
$
}
\caption{Illustration of $\trans{\cdot}: \Signals^\ast \to (\Signals')^\ast$\label{fig:transl}}
\end{figure}
}

In this section, we show decidability for the network
model $\Nmodel = \{\leftrightnet, \leftnet, \rightnet \}$, with
input alphabet $\Signals = \Inpp{1} \times \Nmodel \times \Inpp{2}$
and output alphabet $\Outputs = \Out_1 \times \Out_2$.
Recall that this also implies decidability for the network model
$\{\leftnet,\rightnet\}$. \figtranslation

The idea is to reduce the problem to the case of the network model
$\Nmodel' = \{\leftrightnet, \leftnet \}$ that we considered
in Sections~\ref{sec:onedirectional-model} and
\ref{sec:games}, choosing as input alphabet
$\Signals' = \nInpp{1} \times \Nmodel' \times \nInpp{2}$
where $\nInpp{1} = \nInpp{2} = (\Inpp{1} \cup \Inpp{2}) \uplus \{\dsignal\}$,
and as output alphabet $\Outputs' = \Out_1' \times \Out_2'$
where
$\Out_1' = \Out_2' = (\Out_1 \cup \Out_2) \uplus \{\doutput\}$.
To do so, we will rewrite the given specification $\phi \in \LTL{\Nmodel}$
towards an (automata-based) specification over $\Nmodel'$
in such a way that process 1 can always simulate the ``more informed''
process and process 2 simulates the other process.
Roughly speaking, what we are looking for is a translation
$\trans{\cdot}: \Signals^\ast \to (\Signals')^\ast$
of histories $w$ over $\Nmodel$ to histories $\trans{w}$ over $\Nmodel'$
such that
the view of process 1 in $\trans{w}$ is ``congruent'' to the view
of the more informed process in $w$, and
the view of process 2 in $\trans{w}$ is ``congruent'' to the view of the less informed
process in $w$.
Note that \cite{berwanger2018hierarchical} also
uses a simulation technique to cope with dynamically changing
hierarchies.

\begin{example}
Before defining $\trans{\cdot}$ formally, we illustrate it
in Figure~\ref{fig:transl} for a history $w$.
Round $0$ uses $\leftnet$ so that there is nothing to change.
Round $1$ employs $\rightnet$ so that process 1 henceforth simulates process 2
and vice versa. To make sure that the corresponding views in
$\trans{w}$ are still ``congruent'', we
insert the dummy signal $\inpsymb{\dsignal}{\leftrightnet}{\dsignal}$.
Actually, the gray-shaded view of process 1 in $w$ after round 2 contains the same
information as the gray-shaded view of process 2 in $\trans{w}$ after round $3$.
Though $w$ encounters $\leftrightnet$ in round $3$, we
decide not to change roles again; we will only do so when facing another $\leftnet$
(like in round $5$).
\exend
\end{example}

\newcommand\myoverset[2]{\overset{\textstyle #1}{#2}}

\newcommand{\inout}[5]{
\myoverset{\inpsymb{#1}{#2}{#3}}{\textcolor{red}{(#4,#5)}}
}

\newcommand{\translab}[2]{
\begin{array}{c|c}
#1 & #2
\end{array}
}

\tikzset{
        ->,  
        >=latex, 
        node distance=7cm, 
        initial text=$ $, 
        every loop/.style={looseness=16},
        }
\tikzstyle{every state}=[minimum size=22pt,inner sep=2pt]

\newcommand{\transf}{\delta_{\trans{\cdot}}}

Formally, $\trans{\cdot}: \Signals^\ast \to (\Signals')^\ast$ is given by the sequential
\emph{transducer} shown in Figure~\ref{fig:transldef}.
For the moment, we ignore the red part.
A transition with label $\alpha \mid \beta$ reads $\alpha$ and transforms
it into~$\beta$. As the transducer is deterministic,
it actually defines a function. When we include the red part, i.e.,
the symbols from $\Omega$ and $\Omega'$, we obtain an extension
to $\trans{\cdot}: (\Signals \times \Outputs)^\ast \to (\Signals' \times \Outputs')^\ast$.
Finally, these mappings are extended to
infinite words as expected.

\begin{figure}[h]
\centering
{\small
{
\begin{tikzpicture}
\centering
    \node[state, initial] (s0) {$1$};
    \node[state, right of=s0] (s1) {$2$};

    \draw (s0) edge[bend left=10] node[above]{
    $\translab{\inout{x_1}{\rightnet}{x_2}{y_1}{y_2}}{\inout{\dsignal}{\leftrightnet}{\dsignal}{\dsignal}{\dsignal} \inout{x_2}{\leftnet}{x_1}{y_2}{y_1}}$} (s1);
    \draw (s1) edge[bend left=10] node[below]{$\translab{\inout{x_1}{\leftnet}{x_2}{y_1}{y_2}}{\inout{\dsignal}{\leftrightnet}{\dsignal}{\dsignal}{\dsignal} \inout{x_1}{\leftnet}{x_2}{y_1}{y_2}}$} (s0);

    \draw (s0) edge[loop above] node{$\translab{\inout{x_1}{\leftrightnet}{x_2}{y_1}{y_2}}{\inout{x_1}{\leftrightnet}{x_2}{y_1}{y_2}}$} ();

    \draw (s0) edge[loop below] node{$\translab{\inout{x_1}{\leftnet}{x_2}{y_1}{y_2}}{\inout{x_1}{\leftnet}{x_2}{y_1}{y_2}}$} ();

\draw (s1) edge[loop above] node{$\translab{\inout{x_1}{\leftrightnet}{x_2}{y_1}{y_2}}{\inout{x_2}{\leftrightnet}{x_1}{y_2}{y_1}}$} ();

\draw (s1) edge[loop below] node{$\translab{\inout{x_1}{\rightnet}{x_2}{y_1}{y_2}}{\inout{x_2}{\leftnet}{x_1}{y_2}{y_1}}$} ();

\end{tikzpicture}}}
\caption{The mappings
$\trans{\cdot}: \Signals^\ast \to (\Signals')^\ast$ and
$\trans{\cdot}: (\Signals \times \Outputs)^\ast \to (\Signals' \times \Outputs')^\ast$
\label{fig:transldef}}
\end{figure}

Observe that the state of the transducer
reached after reading $w \in \Signals^\ast$ (or $w \in (\Signals \times \Omega)^\ast$)
reveals the process that process~1 is currently simulating.
We denote this process by 
$\simp{1}(w)$. Accordingly,
$\simp{2}(w) = 3-\simp{1}(w)$
is the process that process 2 simulates after input sequence $w$.
For the example word $w$ in Figure~\ref{fig:transl},
we get $\MOD{1}(w) = 1$ and $\MOD{2}(w) = 2$.

Note that, for all $w,w' \in \Signals^\ast$ and $p \in \{1,2\}$,
such that $\obssyn{p}{w} = \obssyn{p}{w'}$, we have $\simp{p}(w) = \simp{p}(w')$.
This is because the simulated process only depends on the sequence of
links.

Note that the mappings $\trans{\cdot}$ are all injective.
Indeed, at the first position that distinguishes $w$ and $w'$, the transducer
produces letters that distinguish $\trans{w}$ and $\trans{w'}$.
There is an analogous statement for views
(proved in \cite{abs-2002-07545}):

\begin{lemma}\label{lem:injective}
For all $w,w' \in \Signals^\ast$ and $p \in \{1,2\}$, the following hold:
\begin{itemize}
\item[(a)]  $\obssyn{p}{\trans{w}} = \obssyn{p}{\trans{w'}}  ~\Longrightarrow~ \obssyn{\simp{p}(w)}{w} = \obssyn{\simp{p}(w')}{w'}$

\item[(b)] $\obssyn{p}{w} = \obssyn{p}{w'}  ~\Longrightarrow~ \obssyn{\simp{p}(w)}{\trans{w}} = \obssyn{\simp{p}(w')}{\trans{w'}}$
\end{itemize}
\end{lemma}

\medskip

Moreover, the transducer can be applied to $\omega$-regular languages
in the following sense:

\begin{lemma}\label{lem:auttrans}
Given a \DRWA $\A$ over the alphabet $\Signals \times \Outputs$, there
is a \DRWA $\A'$ over $\Signals' \times \Outputs'$ of linear size such that
$L(\A') = \trans{L(\A)} \df \{\trans{w} \mid w \in L(\A)\}$.
\end{lemma}

Now, decidability for $\Nmodel$ is due to
Lemma~\ref{lem:auttrans} and
the following result,
whose proof crucially relies on injectivity of $\trans{\cdot}$ and Lemma~\ref{lem:injective}
(cf.\ \cite{abs-2002-07545}):

\begin{lemma}\label{lem:simulation}
Let $\varphi \in \LTL{\Nmodel}$.
The following statements are equivalent:
\begin{enumerate}
\item[(i)] There is a distributed algorithm $\profile$ (over $\Nmodel$) such that,
for all $w \in \Signals^\omega$, $\outcome{\profile}{w} \in L(\varphi)$.
\item[(ii)] There is a distributed algorithm $\profile'$ (over $\Nmodel'$) such that,
for all $w \in \Signals^\omega$, $\outcome{\profile'}{\trans{w}} \in \trans{L(\varphi)}$.
\end{enumerate}
\end{lemma}

In other words, an instance $\phi \in \LTL{\Nmodel}$ of the synthesis problem
can be reduced to the existence of a distributed
algorithm $f'$ over $\Nmodel'$, $\Signals'$, and $\Outputs'$ that fulfills
$L=M \cup \trans{L(\varphi)}$ where $M \subseteq (\Signals' \times \Omega')^\omega$ is the set of words
whose projection to $\Signals'$ is \emph{not} contained in $\trans{\Signals^\omega}$.
Using Lemma~\ref{lem:auttrans}, we obtain a \DRWA for $L$ (of doubly
exponential size) so that, by Corollary~\ref{cor:autsynth}, the problem is decidable.
Again, the overall procedure runs in 4-fold exponential time.

\smallskip

This concludes the proof of our main result, Theorem~\ref{thm:decidable}.

\section{Conclusion}
\label{sec:conclusion}

We showed that synthesis in a dynamic, synchronous two-node system
  is decidable for LTL specifications if and only if
  the network model does not contain the empty network.
Our model covers full-information protocols where nodes communicate their
  complete unbounded causal history.

Future work is concerned with establishing the precise complexity
of our problem and, possibly, improving the 4-fold exponential
upper bound.
Moreover, it would
be interesting to identify the subsets of
$\{\emptynet,\leftnet,\rightnet,\leftrightnet\}^\omega$ that
give rise to a decidable synthesis problem. For example, one may allow
boundedly many empty links in an input sequence.
Finally, we plan to extend our model to distributed systems of
arbitrary size. We conjecture that synthesis is solvable over a given network
model if and only if, in each communication graph,
any two nodes are connected via a directed path.
This would yield an analogue of the information-fork criterion
\cite{finkbeiner2005uniform}, which applies to static architectures.
It remains to be seen whether the ideas
presented in \cite{finkbeiner2005uniform} can be lifted
to dynamic architectures with causal memory.

\medskip

\noindent
\textit{Acknowledgments.} We thank Dietmar Berwanger for valuable feedback.
This work was partly supported by ANR FREDDA (ANR-17-CE40-0013).

\bibliographystyle{eptcs}
\bibliography{paper}

\begin{thebibliography}{10}
\providecommand{\bibitemdeclare}[2]{}
\providecommand{\surnamestart}{}
\providecommand{\surnameend}{}
\providecommand{\urlprefix}{Available at }
\providecommand{\url}[1]{\texttt{#1}}
\providecommand{\href}[2]{\texttt{#2}}
\providecommand{\urlalt}[2]{\href{#1}{#2}}
\providecommand{\doi}[1]{doi:\urlalt{http://dx.doi.org/#1}{#1}}
\providecommand{\bibinfo}[2]{#2}

\bibitemdeclare{misc}{I2C}
\bibitem{I2C}
 (\bibinfo{year}{2014}): \emph{\bibinfo{title}{{I2C}-Bus Specification and User
  Manual}}.
\newblock
  \bibinfo{howpublished}{\href{https://www.nxp.com/docs/en/user-guide/UM10204.pdf}{https://www.nxp.com/docs/en/user-guide/UM10204.pdf}}.

\bibitemdeclare{misc}{CAN}
\bibitem{CAN}
 (\bibinfo{year}{2014}): \emph{\bibinfo{title}{{ISO} 11898-1:2015. {R}oad
  vehicles —- {C}ontroller area network (CAN) — Part 1: Data link layer and
  physical signalling}}.
\newblock
  \bibinfo{howpublished}{\href{https://www.iso.org/standard/63648.html}{{https://www.iso.org/standard/63648.html}}}.

\bibitemdeclare{inproceedings}{abadi1989realizable}
\bibitem{abadi1989realizable}
\bibinfo{author}{Mart{\'\i}n \surnamestart Abadi\surnameend},
  \bibinfo{author}{Leslie \surnamestart Lamport\surnameend} \&
  \bibinfo{author}{Pierre \surnamestart Wolper\surnameend}
  (\bibinfo{year}{1989}): \emph{\bibinfo{title}{Realizable and unrealizable
  specifications of reactive systems}}.
\newblock In: {\sl \bibinfo{booktitle}{International Colloquium on Automata,
  Languages, and Programming (ICALP'89)}}, \bibinfo{organization}{Springer},
  pp. \bibinfo{pages}{1--17}, \doi{10.1007/BFb0035748}.

\bibitemdeclare{article}{afek1994reliable}
\bibitem{afek1994reliable}
\bibinfo{author}{Yehuda \surnamestart Afek\surnameend}, \bibinfo{author}{Hagit
  \surnamestart Attiya\surnameend}, \bibinfo{author}{Alan \surnamestart
  Fekete\surnameend}, \bibinfo{author}{Michael \surnamestart
  Fischer\surnameend}, \bibinfo{author}{Nancy \surnamestart Lynch\surnameend},
  \bibinfo{author}{Yishay \surnamestart Mansour\surnameend},
  \bibinfo{author}{Dai-Wei \surnamestart Wang\surnameend} \&
  \bibinfo{author}{Lenore \surnamestart Zuck\surnameend}
  (\bibinfo{year}{1994}): \emph{\bibinfo{title}{Reliable communication over
  unreliable channels}}.
\newblock {\sl \bibinfo{journal}{Journal of the ACM (JACM)}}
  \bibinfo{volume}{41}(\bibinfo{number}{6}), pp. \bibinfo{pages}{1267--1297},
  \doi{10.1145/195613.195651}.

\bibitemdeclare{article}{aho1982bounds}
\bibitem{aho1982bounds}
\bibinfo{author}{Alfred~V. \surnamestart Aho\surnameend},
  \bibinfo{author}{Aaron~D. \surnamestart Wyner\surnameend},
  \bibinfo{author}{Mihalis \surnamestart Yannakakis\surnameend} \&
  \bibinfo{author}{Jeffrey~D. \surnamestart Ullman\surnameend}
  (\bibinfo{year}{1982}): \emph{\bibinfo{title}{Bounds on the size and
  transmission rate of communications protocols}}.
\newblock {\sl \bibinfo{journal}{Computers \& Mathematics with Applications}}
  \bibinfo{volume}{8}(\bibinfo{number}{3}), pp. \bibinfo{pages}{205--214},
  \doi{10.1016/0898-1221(82)90043-8}.

\bibitemdeclare{inproceedings}{akkoyunlu1975some}
\bibitem{akkoyunlu1975some}
\bibinfo{author}{Eralp~A. \surnamestart Akkoyunlu\surnameend},
  \bibinfo{author}{Kattamuri \surnamestart Ekanadham\surnameend} \&
  \bibinfo{author}{Richard~V. \surnamestart Huber\surnameend}
  (\bibinfo{year}{1975}): \emph{\bibinfo{title}{Some constraints and tradeoffs
  in the design of network communications}}.
\newblock In: {\sl \bibinfo{booktitle}{5th ACM symposium on Operating Systems
  Principles}}, pp. \bibinfo{pages}{67--74}, \doi{10.1145/800213.806523}.

\bibitemdeclare{article}{bartlett1969note}
\bibitem{bartlett1969note}
\bibinfo{author}{Keith~A. \surnamestart Bartlett\surnameend},
  \bibinfo{author}{Roger~A. \surnamestart Scantlebury\surnameend} \&
  \bibinfo{author}{Peter~T. \surnamestart Wilkinson\surnameend}
  (\bibinfo{year}{1969}): \emph{\bibinfo{title}{A note on reliable full-duplex
  transmission over half-duplex links}}.
\newblock {\sl \bibinfo{journal}{Communications of the ACM (CACM)}}
  \bibinfo{volume}{12}(\bibinfo{number}{5}), pp. \bibinfo{pages}{260--261},
  \doi{10.1145/362946.362970}.

\bibitemdeclare{misc}{abs-2002-07545}
\bibitem{abs-2002-07545}
\bibinfo{author}{B{\'{e}}atrice \surnamestart B{\'{e}}rard\surnameend},
  \bibinfo{author}{Benedikt \surnamestart Bollig\surnameend},
  \bibinfo{author}{Patricia \surnamestart Bouyer\surnameend},
  \bibinfo{author}{Matthias \surnamestart F{\"{u}}gger\surnameend} \&
  \bibinfo{author}{Nathalie \surnamestart Sznajder\surnameend}
  (\bibinfo{year}{2020}): \emph{\bibinfo{title}{Synthesis in Presence of
  Dynamic Links}}.
\newblock \bibinfo{note}{HAL report hal-02917542. Available at
  \href{https://hal.archives-ouvertes.fr/hal-02917542}{https://hal.archives-ouvertes.fr/hal-02917542}.}

\bibitemdeclare{article}{berwanger2018hierarchical}
\bibitem{berwanger2018hierarchical}
\bibinfo{author}{Dietmar \surnamestart Berwanger\surnameend},
  \bibinfo{author}{Anup~Basil \surnamestart Mathew\surnameend} \&
  \bibinfo{author}{Marie \surnamestart van~den Bogaard\surnameend}
  (\bibinfo{year}{2018}): \emph{\bibinfo{title}{Hierarchical information and
  the synthesis of distributed strategies}}.
\newblock {\sl \bibinfo{journal}{Acta Informatica}}
  \bibinfo{volume}{55}(\bibinfo{number}{8}), pp. \bibinfo{pages}{669--701},
  \doi{10.1007/s00236-017-0306-5}.

\bibitemdeclare{incollection}{buchi1990solving}
\bibitem{buchi1990solving}
\bibinfo{author}{J.~Richard \surnamestart B\"uchi\surnameend} \&
  \bibinfo{author}{Lawrence~H. \surnamestart Landweber\surnameend}
  (\bibinfo{year}{1990}): \emph{\bibinfo{title}{Solving sequential conditions
  by finite-state strategies}}.
\newblock In: {\sl \bibinfo{booktitle}{The Collected Works of J. Richard
  B{\"u}chi}}, \bibinfo{publisher}{Springer}, pp. \bibinfo{pages}{525--541},
  \doi{10.1007/978-1-4613-8928-6\_29}.

\bibitemdeclare{inproceedings}{CFN15:icalp}
\bibitem{CFN15:icalp}
\bibinfo{author}{Bernadette \surnamestart Charron{-}Bost\surnameend},
  \bibinfo{author}{Matthias \surnamestart F{\"{u}}gger\surnameend} \&
  \bibinfo{author}{Thomas \surnamestart Nowak\surnameend}
  (\bibinfo{year}{2015}): \emph{\bibinfo{title}{Approximate Consensus in Highly
  Dynamic Networks: {T}he Role of Averaging Algorithms}}.
\newblock In: {\sl \bibinfo{booktitle}{42nd International Colloquium on
  Automata, Languages, and Programming (ICALP'15)}}, pp.
  \bibinfo{pages}{528--539}, \doi{10.1007/978-3-662-47666-6\_42}.

\bibitemdeclare{article}{charron2009heard}
\bibitem{charron2009heard}
\bibinfo{author}{Bernadette \surnamestart Charron-Bost\surnameend} \&
  \bibinfo{author}{Andr{\'e} \surnamestart Schiper\surnameend}
  (\bibinfo{year}{2009}): \emph{\bibinfo{title}{The heard-of model: computing
  in distributed systems with benign faults}}.
\newblock {\sl \bibinfo{journal}{Distributed Computing}}
  \bibinfo{volume}{22}(\bibinfo{number}{1}), pp. \bibinfo{pages}{49--71},
  \doi{10.1007/s00446-009-0084-6}.

\bibitemdeclare{article}{church1957applications}
\bibitem{church1957applications}
\bibinfo{author}{Alonzo \surnamestart Church\surnameend}
  (\bibinfo{year}{1957}): \emph{\bibinfo{title}{Applications of recursive
  arithmetic to the problem of circuit synthesis~--~Summaries of talks}}.
\newblock {\sl \bibinfo{journal}{Institute for Symbolic Logic, Cornell
  University}}.

\bibitemdeclare{inproceedings}{coulouma2013characterization}
\bibitem{coulouma2013characterization}
\bibinfo{author}{{\'E}tienne \surnamestart Coulouma\surnameend} \&
  \bibinfo{author}{Emmanuel \surnamestart Godard\surnameend}
  (\bibinfo{year}{2013}): \emph{\bibinfo{title}{A Characterization of Dynamic
  Networks Where Consensus is Solvable}}.
\newblock In: {\sl \bibinfo{booktitle}{International Colloquium on Structural
  Information and Communication Complexity (SIROCCO'13)}},
  \bibinfo{organization}{Springer}, pp. \bibinfo{pages}{24--35},
  \doi{10.1007/978-3-319-03578-9\_3}.

\bibitemdeclare{inproceedings}{dimitrova2009synthesis}
\bibitem{dimitrova2009synthesis}
\bibinfo{author}{Rayna \surnamestart Dimitrova\surnameend} \&
  \bibinfo{author}{Bernd \surnamestart Finkbeiner\surnameend}
  (\bibinfo{year}{2009}): \emph{\bibinfo{title}{Synthesis of fault-tolerant
  distributed systems}}.
\newblock In: {\sl \bibinfo{booktitle}{International Symposium on Automated
  Technology for Verification and Analysis (ATVA'09)}},
  \bibinfo{organization}{Springer}, pp. \bibinfo{pages}{321--336},
  \doi{10.1007/978-3-642-04761-9\_24}.

\bibitemdeclare{book}{fagin2003reasoning}
\bibitem{fagin2003reasoning}
\bibinfo{author}{Ronald \surnamestart Fagin\surnameend}, \bibinfo{author}{Yoram
  \surnamestart Moses\surnameend}, \bibinfo{author}{Joseph~Y \surnamestart
  Halpern\surnameend} \& \bibinfo{author}{Moshe~Y. \surnamestart
  Vardi\surnameend} (\bibinfo{year}{2003}): \emph{\bibinfo{title}{Reasoning
  about knowledge}}.
\newblock \bibinfo{publisher}{MIT press}.

\bibitemdeclare{inproceedings}{finkbeiner2005uniform}
\bibitem{finkbeiner2005uniform}
\bibinfo{author}{Bernd \surnamestart Finkbeiner\surnameend} \&
  \bibinfo{author}{Sven \surnamestart Schewe\surnameend}
  (\bibinfo{year}{2005}): \emph{\bibinfo{title}{Uniform distributed
  synthesis}}.
\newblock In: {\sl \bibinfo{booktitle}{20th Annual IEEE Symposium on Logic in
  Computer Science (LICS'05)}}, \bibinfo{organization}{IEEE}, pp.
  \bibinfo{pages}{321--330}, \doi{10.1109/LICS.2005.53}.

\bibitemdeclare{inproceedings}{gastin2004distributed}
\bibitem{gastin2004distributed}
\bibinfo{author}{Paul \surnamestart Gastin\surnameend},
  \bibinfo{author}{Benjamin \surnamestart Lerman\surnameend} \&
  \bibinfo{author}{Marc \surnamestart Zeitoun\surnameend}
  (\bibinfo{year}{2004}): \emph{\bibinfo{title}{Distributed games with causal
  memory are decidable for series-parallel systems}}.
\newblock In: {\sl \bibinfo{booktitle}{International Conference on Foundations
  of Software Technology and Theoretical Computer Science (FSTTCS'04)}},
  \bibinfo{organization}{Springer}, pp. \bibinfo{pages}{275--286},
  \doi{10.1007/978-3-540-30538-5\_23}.

\bibitemdeclare{article}{gastin2009distributed}
\bibitem{gastin2009distributed}
\bibinfo{author}{Paul \surnamestart Gastin\surnameend},
  \bibinfo{author}{Nathalie \surnamestart Sznajder\surnameend} \&
  \bibinfo{author}{Marc \surnamestart Zeitoun\surnameend}
  (\bibinfo{year}{2009}): \emph{\bibinfo{title}{Distributed synthesis for
  well-connected architectures}}.
\newblock {\sl \bibinfo{journal}{Formal Methods in System Design}}
  \bibinfo{volume}{34}(\bibinfo{number}{3}), pp. \bibinfo{pages}{215--237},
  \doi{10.1007/s10703-008-0064-7}.

\bibitemdeclare{inproceedings}{GenestGMW13}
\bibitem{GenestGMW13}
\bibinfo{author}{Blaise \surnamestart Genest\surnameend}, \bibinfo{author}{Hugo
  \surnamestart Gimbert\surnameend}, \bibinfo{author}{Anca \surnamestart
  Muscholl\surnameend} \& \bibinfo{author}{Igor \surnamestart
  Walukiewicz\surnameend} (\bibinfo{year}{2013}):
  \emph{\bibinfo{title}{Asynchronous Games over Tree Architectures}}.
\newblock In: {\sl \bibinfo{booktitle}{Automata, Languages, and Programming -
  40th International Colloquium, {ICALP} 2013, Riga, Latvia, July 8-12, 2013,
  Proceedings, Part {II}}}, {\sl \bibinfo{series}{Lecture Notes in Computer
  Science}} \bibinfo{volume}{7966}, \bibinfo{publisher}{Springer}, pp.
  \bibinfo{pages}{275--286}, \doi{10.1007/978-3-642-39212-2\_26}.

\bibitemdeclare{inproceedings}{gimbert18}
\bibitem{gimbert18}
\bibinfo{author}{Hugo \surnamestart Gimbert\surnameend} (\bibinfo{year}{2018}):
  \emph{\bibinfo{title}{On the Control of Asynchronous Automata}}.
\newblock In: {\sl \bibinfo{booktitle}{37th IARCS Annual Conference on
  Foundations of Software Technology and Theoretical Computer Science
  (FSTTCS'17)}}, {\sl \bibinfo{series}{LIPIcs}}~\bibinfo{volume}{93},
  \bibinfo{publisher}{Schloss Dagstuhl--Leibniz-Zentrum f\"ur Informatik}, pp.
  \bibinfo{pages}{30:1--30:15}, \doi{10.4230/LIPIcs.FSTTCS.2017.30}.

\bibitemdeclare{article}{Klarlund94}
\bibitem{Klarlund94}
\bibinfo{author}{Nils \surnamestart Klarlund\surnameend}
  (\bibinfo{year}{1994}): \emph{\bibinfo{title}{Progress Measures, Immediate
  Determinacy, and a Subset Construction for Tree Automata}}.
\newblock {\sl \bibinfo{journal}{Annals of Pure and Applied Logic}}
  \bibinfo{volume}{69}(\bibinfo{number}{2-3}), pp. \bibinfo{pages}{243--268},
  \doi{10.1016/0168-0072(94)90086-8}.

\bibitemdeclare{inproceedings}{kuhn2010distributed}
\bibitem{kuhn2010distributed}
\bibinfo{author}{Fabian \surnamestart Kuhn\surnameend}, \bibinfo{author}{Nancy
  \surnamestart Lynch\surnameend} \& \bibinfo{author}{Rotem \surnamestart
  Oshman\surnameend} (\bibinfo{year}{2010}): \emph{\bibinfo{title}{Distributed
  computation in dynamic networks}}.
\newblock In: {\sl \bibinfo{booktitle}{42nd ACM Symposium on Theory of
  Computing (STOC'10)}}, pp. \bibinfo{pages}{513--522},
  \doi{10.1145/1806689.1806760}.

\bibitemdeclare{article}{kupferman1999church}
\bibitem{kupferman1999church}
\bibinfo{author}{Orna \surnamestart Kupferman\surnameend} \&
  \bibinfo{author}{Moshe~Y. \surnamestart Vardi\surnameend}
  (\bibinfo{year}{1999}): \emph{\bibinfo{title}{Church's problem revisited}}.
\newblock {\sl \bibinfo{journal}{Bulletin of Symbolic Logic}}
  \bibinfo{volume}{5}(\bibinfo{number}{2}), pp. \bibinfo{pages}{245--263},
  \doi{10.2307/421091}.

\bibitemdeclare{incollection}{kupferman2000synthesis}
\bibitem{kupferman2000synthesis}
\bibinfo{author}{Orna \surnamestart Kupferman\surnameend} \&
  \bibinfo{author}{Moshe~Y. \surnamestart Vardi\surnameend}
  (\bibinfo{year}{2000}): \emph{\bibinfo{title}{Synthesis with incomplete
  information}}.
\newblock In: {\sl \bibinfo{booktitle}{Advances in Temporal Logic}},
  \bibinfo{publisher}{Springer}, pp. \bibinfo{pages}{109--127},
  \doi{10.1007/978-94-015-9586-5\_6}.

\bibitemdeclare{inproceedings}{kupferman2001synthesizing}
\bibitem{kupferman2001synthesizing}
\bibinfo{author}{Orna \surnamestart Kupferman\surnameend} \&
  \bibinfo{author}{Moshe~Y. \surnamestart Vardi\surnameend}
  (\bibinfo{year}{2001}): \emph{\bibinfo{title}{Synthesizing distributed
  systems}}.
\newblock In: {\sl \bibinfo{booktitle}{16th Annual IEEE Symposium on Logic in
  Computer Science (LICS'01)}}, \bibinfo{organization}{IEEE}, pp.
  \bibinfo{pages}{389--398}, \doi{10.1109/LICS.2001.932514}.

\bibitemdeclare{article}{lamport1982byzantine}
\bibitem{lamport1982byzantine}
\bibinfo{author}{Leslie \surnamestart Lamport\surnameend},
  \bibinfo{author}{Robert \surnamestart Shostak\surnameend} \&
  \bibinfo{author}{Marshall \surnamestart Pease\surnameend}
  (\bibinfo{year}{1982}): \emph{\bibinfo{title}{The Byzantine Generals
  Problem}}.
\newblock {\sl \bibinfo{journal}{ACM Transactions on Programming Languages and
  Systems}} \bibinfo{volume}{4}(\bibinfo{number}{3}), pp.
  \bibinfo{pages}{382--401}, \doi{10.1145/357172.357176}.

\bibitemdeclare{book}{lynch1996distributed}
\bibitem{lynch1996distributed}
\bibinfo{author}{Nancy~A. \surnamestart Lynch\surnameend}
  (\bibinfo{year}{1996}): \emph{\bibinfo{title}{Distributed algorithms}}.
\newblock \bibinfo{publisher}{Elsevier}.

\bibitemdeclare{inproceedings}{MadhusudanTY05}
\bibitem{MadhusudanTY05}
\bibinfo{author}{P.~\surnamestart Madhusudan\surnameend}, \bibinfo{author}{P.S.
  \surnamestart Thiagarajan\surnameend} \& \bibinfo{author}{Shaofa
  \surnamestart Yang\surnameend} (\bibinfo{year}{2005}):
  \emph{\bibinfo{title}{The {MSO} theory of connectedly communicating
  processes}}.
\newblock In: {\sl \bibinfo{booktitle}{International Conference on Foundations
  of Software Technology and Theoretical Computer Science (FSTTCS'05)}},
  \bibinfo{organization}{Springer}, pp. \bibinfo{pages}{201--212},
  \doi{10.1007/11590156\_16}.

\bibitemdeclare{inproceedings}{van2005synthesis}
\bibitem{van2005synthesis}
\bibinfo{author}{Ron \surnamestart Van~der Meyden\surnameend} \&
  \bibinfo{author}{Thomas \surnamestart Wilke\surnameend}
  (\bibinfo{year}{2005}): \emph{\bibinfo{title}{Synthesis of distributed
  systems from knowledge-based specifications}}.
\newblock In: {\sl \bibinfo{booktitle}{International Conference on Concurrency
  Theory (CONCUR'05)}}, \bibinfo{organization}{Springer}, pp.
  \bibinfo{pages}{562--576}, \doi{10.1007/11539452\_42}.

\bibitemdeclare{techreport}{van2005synthesis:TR}
\bibitem{van2005synthesis:TR}
\bibinfo{author}{Ron \surnamestart Van~der Meyden\surnameend} \&
  \bibinfo{author}{Thomas \surnamestart Wilke\surnameend}
  (\bibinfo{year}{2005}): \emph{\bibinfo{title}{Synthesis of distributed
  systems from knowledge-based specifications. UNSW-CSE-TR-0504}}.
\newblock \bibinfo{type}{Technical Report}, \bibinfo{institution}{UNSW Sydney}.

\bibitemdeclare{inproceedings}{mohalik2003distributed}
\bibitem{mohalik2003distributed}
\bibinfo{author}{Swarup \surnamestart Mohalik\surnameend} \&
  \bibinfo{author}{Igor \surnamestart Walukiewicz\surnameend}
  (\bibinfo{year}{2003}): \emph{\bibinfo{title}{Distributed games}}.
\newblock In: {\sl \bibinfo{booktitle}{International Conference on Foundations
  of Software Technology and Theoretical Computer Science (FSTTCS'03)}},
  \bibinfo{organization}{Springer}, pp. \bibinfo{pages}{338--351},
  \doi{10.1007/978-3-540-24597-1\_29}.

\bibitemdeclare{inproceedings}{peterson1979multiple}
\bibitem{peterson1979multiple}
\bibinfo{author}{Gary~L. \surnamestart Peterson\surnameend} \&
  \bibinfo{author}{John~H. \surnamestart Reif\surnameend}
  (\bibinfo{year}{1979}): \emph{\bibinfo{title}{Multiple-person alternation}}.
\newblock In: {\sl \bibinfo{booktitle}{20th Annual Symposium on Foundations of
  Computer Science (FOCS'79)}}, \bibinfo{organization}{IEEE}, pp.
  \bibinfo{pages}{348--363}, \doi{10.1109/SFCS.1979.25}.

\bibitemdeclare{article}{pnueli1981temporal}
\bibitem{pnueli1981temporal}
\bibinfo{author}{Amir \surnamestart Pnueli\surnameend} (\bibinfo{year}{1981}):
  \emph{\bibinfo{title}{The temporal semantics of concurrent programs}}.
\newblock {\sl \bibinfo{journal}{Theoretical Computer Science}}
  \bibinfo{volume}{13}(\bibinfo{number}{1}), pp. \bibinfo{pages}{45--60},
  \doi{10.1016/0304-3975(81)90110-9}.

\bibitemdeclare{inproceedings}{pnueli1988framework}
\bibitem{pnueli1988framework}
\bibinfo{author}{Amir \surnamestart Pnueli\surnameend} \& \bibinfo{author}{Roni
  \surnamestart Rosner\surnameend} (\bibinfo{year}{1988}):
  \emph{\bibinfo{title}{A framework for the synthesis of reactive modules}}.
\newblock In: {\sl \bibinfo{booktitle}{International Conference on Concurrency
  (Concurrency 88)}}, \bibinfo{organization}{Springer}, pp.
  \bibinfo{pages}{4--17}, \doi{10.1007/3-540-50403-6\_28}.

\bibitemdeclare{inproceedings}{pnueli1989synthesis}
\bibitem{pnueli1989synthesis}
\bibinfo{author}{Amir \surnamestart Pnueli\surnameend} \& \bibinfo{author}{Roni
  \surnamestart Rosner\surnameend} (\bibinfo{year}{1989}):
  \emph{\bibinfo{title}{On the synthesis of a reactive module}}.
\newblock In: {\sl \bibinfo{booktitle}{16th ACM SIGPLAN-SIGACT symposium on
  Principles of Programming Languages (POPL'89)}}, pp.
  \bibinfo{pages}{179--190}, \doi{10.1145/75277.75293}.

\bibitemdeclare{inproceedings}{pnueli1990distributed}
\bibitem{pnueli1990distributed}
\bibinfo{author}{Amir \surnamestart Pnueli\surnameend} \& \bibinfo{author}{Roni
  \surnamestart Rosner\surnameend} (\bibinfo{year}{1990}):
  \emph{\bibinfo{title}{Distributed reactive systems are hard to synthesize}}.
\newblock In: {\sl \bibinfo{booktitle}{31st Annual Symposium on Foundations of
  Computer Science (FoCS'90)}}, \bibinfo{organization}{IEEE}, pp.
  \bibinfo{pages}{746--757}, \doi{10.1109/FSCS.1990.89597}.

\bibitemdeclare{book}{rabin1972automata}
\bibitem{rabin1972automata}
\bibinfo{author}{Michael~O. \surnamestart Rabin\surnameend}
  (\bibinfo{year}{1972}): \emph{\bibinfo{title}{Automata on infinite objects
  and Church's problem}}.
\newblock \bibinfo{volume}{13}, \bibinfo{publisher}{American Mathematical
  Soc.}, \doi{10.1090/cbms/013}.

\bibitemdeclare{inproceedings}{Safra88}
\bibitem{Safra88}
\bibinfo{author}{Shmuel \surnamestart Safra\surnameend} (\bibinfo{year}{1988}):
  \emph{\bibinfo{title}{On the Complexity of $\omega$-Automata}}.
\newblock In: {\sl \bibinfo{booktitle}{29th Annual Symposium on Foundations of
  Computer Science (FoCS'88)}}, \bibinfo{publisher}{{IEEE} Computer Society},
  pp. \bibinfo{pages}{319--327}, \doi{10.1109/SFCS.1988.21948}.

\bibitemdeclare{incollection}{Thomas90}
\bibitem{Thomas90}
\bibinfo{author}{Wolfgang \surnamestart Thomas\surnameend}
  (\bibinfo{year}{1990}): \emph{\bibinfo{title}{Automata on Infinite Objects}}.
\newblock In \bibinfo{editor}{Jan \surnamestart van Leeuwen\surnameend},
  editor: {\sl \bibinfo{booktitle}{Handbook of Theoretical Computer Science,
  Volume {B:} Formal Models and Semantics}}, \bibinfo{publisher}{Elsevier and
  {MIT} Press}, pp. \bibinfo{pages}{133--191},
  \doi{10.1016/B978-0-444-88074-1.50009-3}.

\bibitemdeclare{inproceedings}{thomas1995synthesis}
\bibitem{thomas1995synthesis}
\bibinfo{author}{Wolfgang \surnamestart Thomas\surnameend}
  (\bibinfo{year}{1995}): \emph{\bibinfo{title}{On the synthesis of strategies
  in infinite games}}.
\newblock In: {\sl \bibinfo{booktitle}{Annual Symposium on Theoretical Aspects
  of Computer Science (STACS'95)}}, \bibinfo{organization}{Springer}, pp.
  \bibinfo{pages}{1--13}, \doi{10.1007/3-540-59042-0_57}.

\bibitemdeclare{article}{VardiW94}
\bibitem{VardiW94}
\bibinfo{author}{Moshe~Y. \surnamestart Vardi\surnameend} \&
  \bibinfo{author}{Pierre \surnamestart Wolper\surnameend}
  (\bibinfo{year}{1994}): \emph{\bibinfo{title}{Reasoning About Infinite
  Computations}}.
\newblock {\sl \bibinfo{journal}{Information and Computation}}
  \bibinfo{volume}{115}(\bibinfo{number}{1}), pp. \bibinfo{pages}{1--37},
  \doi{10.1006/inco.1994.1092}.

\bibitemdeclare{inproceedings}{velner2011church}
\bibitem{velner2011church}
\bibinfo{author}{Yaron \surnamestart Velner\surnameend} \&
  \bibinfo{author}{Alexander \surnamestart Rabinovich\surnameend}
  (\bibinfo{year}{2011}): \emph{\bibinfo{title}{Church synthesis problem for
  noisy input}}.
\newblock In: {\sl \bibinfo{booktitle}{International Conference on Foundations
  of Software Science and Computational Structures (FoSSaCS'11)}},
  \bibinfo{organization}{Springer}, pp. \bibinfo{pages}{275--289},
  \doi{10.1007/978-3-642-19805-2_19}.

\end{thebibliography}

\end{document}